\documentclass[11pt,a4paper]{article}
\usepackage{setspace}
\usepackage[normalem]{ulem}
\usepackage[margin=1in]{geometry}
\usepackage[toc,page]{appendix}
\usepackage{bbold} 
\usepackage{caption}
\usepackage{amsmath, amsfonts, amssymb, amsthm, natbib}
\usepackage{graphicx, algorithm, algorithmicx, enumerate, enumitem, color, hyperref}
\usepackage{algpseudocode, caption}

\newcommand{\real}{\mathbb R}

\def\E{\operatorname{E}}
\def\Var{\operatorname{var}}
\def\Prob{\operatorname{pr}}
\def\tr{\operatorname{tr}}
\def\sign{\operatorname{sign}}

\def\snaive{\hat{\sigma}^2_{\mathrm{naive}}}
\def\L{L}
\def\0{0}
\def\y{y}
\def\X{X}

\def\z{z}
\def\soft{\mathcal{S}}
\def\r{r}

\def\rt{\hat\sigma_{R}^2}

\def\subto{\mathrm{\quad s.t.\quad}}

\def\E{\operatorname{E}}
\def\Var{\operatorname{var}}
\def\Prob{\operatorname{pr}}
\def\tr{\operatorname{tr}}
\def\sign{\operatorname{sign}}

\def\iden{I}
\def\d{\mathrm{d}}

\def\ob{\check{\beta}}
\def\tb{\beta^\ast}
\def\tg{\gamma^\ast}
\def\T{\mathcal{T}}
\def\S{\mathcal{S}}
\def\be{\hat{\beta}_\lambda}
\def\bt{\beta^\ast}
\def\se{\hat{\sigma}^2_{\mathrm{naive}}}

\def\bsqrt{\tilde{\beta}_{\mathrm{SQRT}}}
\def\ssqrt{\tilde{\sigma}^2_{\mathrm{SQRT}}}

\newcommand{\snorm}[1]{\|#1\|}
\newcommand{\norm}[1]{\left\lVert#1\right\rVert}

\DeclareMathOperator*{\argmin}{arg\,min}

\newtheorem{theorem}{Theorem}
\newtheorem{lemma}[theorem]{Lemma}
\newtheorem{proposition}[theorem]{Proposition}
\newtheorem{remark}[theorem]{Remark}
\newtheorem{corollary}[theorem]{Corollary}

\title{Estimating the error variance in a high-dimensional linear model}
\author{Guo Yu\thanks{Department of Statistics, University of Washington, Seattle, Washington, 98105, \href{mailto:gy63@uw.edu}{gy63@uw.edu} } \and Jacob Bien\thanks{Data Sciences and Operations, Marshall School of Business, University of Southern California, Los Angeles, CA 90089, \href{mailto:jbien@usc.edu}{jbien@usc.edu}}}
\date{}

\begin{document}
\maketitle

\begin{abstract}
The lasso has been studied extensively as a tool for estimating the
coefficient vector in the high\verb+-+dimensional linear model; however,
considerably less is known about estimating the error
variance in this context. 
In this paper, we propose the natural lasso estimator for the error
variance, which maximizes a penalized likelihood objective. A key
aspect of the natural lasso is that the likelihood is expressed in
terms of the natural parameterization of the multiparameter exponential
family of a Gaussian with unknown mean and variance. The result is a remarkably simple estimator of the error variance with provably good performance in terms of mean squared error. These theoretical results do not require placing any assumptions on the design matrix or the true regression coefficients. We also propose a companion estimator, called the organic lasso, which theoretically does not require tuning of the regularization parameter. Both estimators do well empirically compared to preexisting methods, especially in settings where successful recovery of the true support of the coefficient vector is hard.
Finally, we show that existing methods can do well under fewer assumptions than previously known, thus providing a fuller story about the problem of estimating the error variance in high\verb+-+dimensional linear models.
\end{abstract}

\section{Introduction}
The linear model 
\begin{align}
  \y = \X \beta^\ast + \mathbf{\varepsilon} \qquad  \varepsilon \sim N(0, \sigma^2 I_n),
  \label{est:model}
\end{align}
is one of the most fundamental models in statistics.
It describes the relationship between a response vector $\y \in
\real^n$ and a fixed design matrix $\X \in \real^{n\times p}$.
When $p \gg n$, estimating the coefficient vector $\beta^\ast$ is a challenging, 
well\verb+-+studied problem.
Perhaps the most common method in this setting is the lasso \citep{tibshirani1996regression},
which assumes that $\beta^\ast$ is sparse and solves the following convex optimization problem: 
\begin{align}
  \hat{\beta}_\lambda \in \argmin_{\beta \in \real^p} \left( \frac{1}{n} \norm{\y - \X\beta}_2^2 + 2\lambda \norm{\beta}_1 \right).
  \label{est:lasso}
\end{align}
Over the past decade, an extensive literature has emerged studying the properties of $\hat{\beta}_\lambda$ from
both computational \citep[e.g.,][]{hastie2015statistical} and
theoretical \citep[e.g.,][]{buhlmann2011statistics} perspectives.

Compared to the vast amount of work on estimating $\beta^\ast$,
relatively little attention has been paid to the problem of estimating
$\sigma^2$, which captures the noise level or extent to which $\y$ cannot be predicted from $\X$.
Nonetheless, reliable estimation of $\sigma^2$ is important for quantifying the uncertainty in estimating $\beta^\ast$.
A series of recent advances in high\verb+-+dimensional inference
\citep[][etc.]{buhlmann2013statistical, zhang2014confidence, van2014asymptotically, lockhart2014significance, javanmard2014confidence,
  lee2016exact, tibshirani2016exact, taylor2017post, ning2017} may very well be the determining factor for the widespread adoption of the lasso and related methods in fields where $p$\verb+-+values and confidence intervals are required.
Thus, estimating $\sigma^2$ reliably in finite samples is crucial.

If $\beta^\ast$ were known, then the optimal estimator for $\sigma^2$ would of course be $n^{-1}\snorm{\y - \X\beta^\ast}^2_2 =n^{-1}\snorm{\varepsilon}_2^2$.
Thus, a naive estimator for $\sigma^2$ based on an estimator $\hat{\beta}$ of $\beta^\ast$ would be
\begin{align}
  \snaive = \frac{1}{n}\snorm{\y - \X \hat{\beta}}^2_2.
  \label{est:naive}
\end{align}
However, a simple calculation in the classical $n > p$ setting shows that such an estimator is biased downward:
a least\verb+-+squares oracle with knowledge of the true support $S = \{ j: \beta^\ast_j \neq 0 \}$ scales this to give an unbiased estimator,
\begin{align}
  \hat{\sigma}^2_{\mathrm{oracle}} = \frac{1}{n - \left| S\right|} \snorm{\y - \X_{S} \X_S^+ \y}_2^2,
  \label{est:oracle_LS}
\end{align}
where $\X_{S}$ is a sub\verb+-+matrix of $\X$ with columns indexed by $S$ and $\X_{S}^+$ is its pseudoinverse.
Many papers in this area discuss the difficulty of estimating $\sigma^2$ and warn of the perils of underestimating it:
if $\sigma^2$ is underestimated then one gets anti\verb+-+conservative confidence intervals, which are highly undesirable \citep{tibshirani2015uniform}.

\citet{reid2013study} carry out an extensive review and simulation study of several estimators of $\sigma^2$ \citep{fan2012variance, Sun25092012, Dicker14MM}, and they devote special attention to studying the estimator
\begin{align}
  \rt = \frac{1}{n - \hat{s}_\lambda} \snorm{\y - \X \hat{\beta}_\lambda}_2^2,
  \label{est:reid}
\end{align}
where $\hat{\beta}_\lambda$ is as in \eqref{est:lasso}, with $\lambda$ selected using a cross\verb+-+validation procedure, and $\hat{s}_\lambda$ is the number of nonzero elements in $\hat{\beta}_\lambda$.
They show that \eqref{est:reid} has promising performance in a wide range of simulation settings 
and provide an asymptotic theoretical understanding of the estimator in the special case where $\X$ is an orthogonal matrix.

While intuition from \eqref{est:oracle_LS} suggests that \eqref{est:reid} is a quite reasonable estimator when $S$ can be well recovered, it also points to the question of how well the estimator will perform when $S$ is not well recovered by the lasso.  The conditions required for the lasso to recover $S$ are much stricter than the conditions
needed for it to do well in prediction \citep[e.g.,][]{Van2009On}.
The scale factor $(n-\hat s_\lambda)^{-1}$ used in $\rt$ means that this approach depends not just on the predicted values of the lasso,
$\X\hat\beta_\lambda$, but on the magnitude of the set of nonzero elements in $\hat\beta_\lambda$. 
Indeed, we find that in situations where recovering $S$ is challenging, $\rt$ tends to yield less favorable empirical performance. The theoretical development in \citet{reid2013study} sidesteps this complication by working in an asymptotic regime in which $\rt$ behaves like the naive estimator \eqref{est:naive}. 
To understand the finite\verb+-+sample performance of $\rt$ would require
considering the behavior of the random variable $\hat s_\lambda$.
Clearly, when $\hat s_\lambda\approx n$, even small fluctuations in
$\hat s_\lambda$ can lead to large fluctuations in $\rt$. 
Finally, from a practical standpoint, computing $\hat s_\lambda$ is a numerically sensitive operation in that it requires the choice of a threshold size for calling a value numerically zero, and the assurance that one has solved the problem to sufficient precision.

Based on these observations, we propose in this paper a completely different approach to estimating $\sigma^2$.
The basic premise of our framework is that when both $\beta^\ast$ and $\sigma^2$ are unknown, it is convenient to formulate the penalized log\verb+-+likelihood problem in terms of
\begin{align}
  \phi = \frac{1}{\sigma^2}, \qquad \theta = \frac{\beta}{\sigma^2},
  \label{est:naturalpara}
\end{align}
the natural parameters of the Gaussian multiparameter exponential family with unknown mean and variance.
The negative Gaussian log\verb+-+likelihood is not jointly convex in the $(\beta,\sigma)$ parameterization. In fact, even with $\beta$ fixed, it is nonconvex in $\sigma$. However, in the natural parameterization the negative log\verb+-+likelihood is jointly convex in $(\phi,\theta)$.

We penalize this negative log\verb+-+likelihood with an $\ell_1$\verb+-+norm on the
natural parameter $\theta$ and call this new estimator the natural lasso.  We show in Section \ref{sec:nlasso} that the resulting error variance estimator can in fact be very simply expressed as the minimizing value of the regular lasso problem \eqref{est:lasso}:
\begin{align}
  \hat{\sigma}^2_\lambda = \min_{\beta \in \real^p} \left( \frac{1}{n} \norm{\y - \X\beta}_2^2 + 2 \lambda \norm{\beta}_1 \right).
  \label{est:sig_nlasso}
\end{align}
Observing that the first term is $\hat{\sigma}^2_{\mathrm{naive}}$, we directly see that the natural lasso counters the naive method's downward bias through an additive correction; this is in contrast to $\rt$'s reliance on a multiplicative correction that sometimes may be unstable.
Computing \eqref{est:sig_nlasso} is clearly no harder than solving a lasso and, unlike $\rt$, does not require determining a threshold for deciding which coefficient estimates are  numerically zero.
Furthermore, we establish finite\verb+-+sample bounds on the mean squared error that hold without making any assumptions on the design matrix $\X$. Our theoretical analysis suggests a second approach that is also based on the natural parameterization.
The theory that we develop for this method, which we call the organic lasso, relies on weaker assumptions.
We find that both methods have competitive empirical performance relative to $\rt$ and show particular strength in settings in which support recovery is known to be challenging.

Our final contribution is to show that existing methods
can also attain high\verb+-+dimensional consistency under no assumptions on $\X$.  In
particular, we provide finite\verb+-+sample bounds for $\hat{\sigma}^2_{\mathrm{naive}}$, with
$\hat{\beta}$ in \eqref{est:naive} taken to be the standard lasso or
the square\verb+-+root/scaled lasso estimator
\citep{belloni2011square, Sun25092012}.  Previous results about
$\hat{\sigma}^2_{\mathrm{naive}}$ have placed strong assumptions on $\X$.
Thus, our work provides a fuller story about the problem of estimating the error variance in high\verb+-+dimensional linear models.

\section{Natural parameterization}
The negative log\verb+-+likelihood function in \eqref{est:model} is, up to a constant,
\begin{align}
  \L\left( \beta, \sigma^2 | \X, \y \right) = \frac{n}{2} \log \sigma^2 + \frac{\norm{\y - \X \beta}_2^2}{2 \sigma^2}.
  \nonumber
\end{align}
When $\sigma^2$ is known, the $\sigma$ dependence can be ignored, leading to the standard least\verb+-+squares criterion;
however, when $\sigma$ is unknown, performing a full minimization of the penalized negative log\verb+-+likelihood amounts to solving a nonconvex optimization problem even with a convex penalty.

The nonconvexity of the Gaussian negative log\verb+-+likelihood in its
variance, or more generally, covariance matrix, is a
well\verb+-+known difficulty \citep{bien2011sparse}.  In this context,
working instead with the inverse covariance matrix is common
\citep{yuan2007model,banerjee2008model,friedman2008sparse}.  We take
an analogous approach when estimation of $\sigma^2$ is of interest, considering the natural parameterization \eqref{est:naturalpara} of the Gaussian multiparameter exponential family with unknown variance,
\begin{align}
  \L\left( \phi^{-1}\theta,\phi^{-1} | \X, \y \right) 
  = -\frac{n}{2} \log \phi + \frac{1}{2} \phi \norm{\y - \X \frac{\theta}{\phi}}_2^2
  = -\frac{n}{2} \log \phi+\phi \frac{\norm{y}_2^2}{2} - \y^T \X \theta + \frac{\norm{\X\theta}_2^2}{2\phi}.
  \nonumber
\end{align}
Observing that attaining sparsity in $\theta$ is equivalent to attaining sparsity in $\beta$, we propose the following penalized maximum log\verb+-+likelihood estimator:
\begin{align}
  \left( \hat{\theta}_\lambda, \hat{\phi}_\lambda \right) 
  \in \argmin_{\phi > 0,~\theta} \left\{ -\frac{1}{2} \log \phi + \phi \frac{\norm{y}_2^2}{2n} - \frac{1}{n} \y^T \X \theta + \frac{\norm{\X\theta}_2^2}{2n\phi} + \lambda \Omega(\theta, \phi) \right\}
  \label{est:naturalest}
\end{align}
for a convex penalty $\Omega(\theta, \phi)$ that induces sparsity in $\theta$. We will focus on $\Omega(\theta, \phi) = \snorm{\theta}_1$ in Section \ref{sec:nlasso} and $\Omega(\theta, \phi) = \phi^{-1} \snorm{\theta}_1^2$ in Section \ref{sec:olasso}.
This problem is jointly convex in $(\theta, \phi)$.  While this is a general property of exponential families due to the convexity of the cumulant generating function, we can see it in this special case because of the convexity of $-\log$ and the convexity of the quadratic\verb+-+over\verb+-+linear function \citep{boyd2004convex, rockafellar2015convex}.  Given a solution to \eqref{est:naturalest}, we can reverse \eqref{est:naturalpara} to get estimators for $\sigma^2$ and $\beta^\ast$:
\begin{align}
  \tilde\sigma^2_\lambda = \frac{1}{\hat\phi_\lambda},\qquad \tilde\beta_\lambda= \frac{\hat\theta_\lambda}{\hat\phi_\lambda}.
  \label{eq:nl}
\end{align}

Before proceeding with an analysis of the estimator \eqref{eq:nl} with specific choices of $\Omega(\theta, \phi)$, we point out a similarity between our method and that of \citet{st2010l1}, who consider a different convexifying reparameterization of the Gaussian log\verb+-+likelihood, using
$\rho = \sigma^{-1}$ and $\gamma = \sigma^{-1}\beta$.  They put an $\ell_1$\verb+-+norm penalty on $\gamma$, which has the same sparsity pattern as $\beta$, and solve
\begin{align}
  \min_{\rho>0,\gamma} \left( -\log \rho + \frac{1}{2n} \norm{\rho \y - \X \gamma}_2^2 + \lambda \norm{\gamma}_1 \right).
  \label{est:stadler}
\end{align}
\citet{sun2010comments} give an asymptotic analysis of the solution to \eqref{est:stadler} under a compatibility condition.
A modification of this problem \citep{antoniadis2010comments} gives the scaled lasso \citep{Sun25092012}, which is known to be equivalent to the square\verb+-+root lasso \citep{belloni2011square}:
\begin{align}
  \tilde{\beta}_{\mathrm{SQRT}} = \argmin_{\beta \in \real^p} \left( \frac{1}{\surd{n}} \norm{\y - \X \beta}_2 + \lambda \norm{\beta}_1 \right),
  \qquad
  \tilde{\sigma}^2_{\mathrm{SQRT}} = \frac{1}{n} \norm{\y - \X \tilde{\beta}_{\mathrm{SQRT}}}_2^2.
  \label{est:sqrt}
\end{align}
With the same parameterization $(\rho, \gamma)$, \citet{dalalyan2012fused} propose the scaled Dantzig selector under the assumption of fused sparsity. Under the restricted eigenvalue condition, they establish the same rate of convergence in estimating the error variance as the fast prediction error rate of the standard lasso \citep{hebiri2013correlations, 2016arXiv160800624L, dalalyan2017}. 

\section{The natural lasso estimator of error variance} \label{sec:nlasso}
We first propose the natural lasso, which is the solution to \eqref{est:naturalest} with $\Omega(\theta, \phi) = \snorm{\theta}_1$.
One might think that solving the natural lasso
would involve a specialized algorithm.  The
following proposition shows, remarkably, that this is not the case.

\begin{proposition}
  The natural lasso estimator $(\tilde\beta_\lambda,\tilde\sigma^2_\lambda)$ defined in \eqref{eq:nl}, where $(\hat{\theta}_\lambda, \hat{\phi}_\lambda)$ is a solution to \eqref{est:naturalest} with $\Omega(\theta, \phi) = \snorm{\theta}_1$, satisfies the following properties:
  \begin{enumerate}
    \item $\tilde\beta_\lambda=\hat\beta_\lambda$, a solution to the standard lasso \eqref{est:lasso};
    \item $\tilde\sigma^2_\lambda=\hat\sigma^2_\lambda$, the standard lasso's optimal value \eqref{est:sig_nlasso}.
  \end{enumerate}
  Furthermore, 
  $\hat\sigma^2_\lambda= n^{-1}(\snorm{\y}_2^2 - \snorm{\X \hat{\beta}_\lambda}_2^2)$.
  \label{main:prop:nl}
\end{proposition}
The proof of this proposition and all theoretical results that follow
can be found in the Appendices.
Thus, to get the natural lasso estimator of $(\beta^\ast,\sigma^2)$, one simply solves the standard lasso \eqref{est:lasso} and returns a solution and the minimal value.  

An attractive property of the natural lasso estimator
$\hat\sigma_\lambda^2$ is the relative ease with which one can prove
bounds about its performance.  Since $\hat\sigma_\lambda^2$ is the
optimal value of the lasso problem, the objective value at any vector
$\beta$ provides an upper bound on $\hat\sigma_\lambda^2$.  Likewise,
any dual feasible vector provides a lower bound on
$\hat\sigma_\lambda^2$.  These considerations are used to prove the following lemma, which shows that for a suitably chosen $\lambda$, the natural lasso variance estimator gets close to the oracle estimator of $\sigma^2$.
\begin{lemma} \label{lem1}
  If $\lambda \geq n^{-1}\snorm{\X^T \varepsilon}_\infty$, then $| \hat\sigma^2_\lambda - n^{-1}\snorm{\varepsilon}^2_2 | \le 2\lambda\snorm{\beta^\ast}_1$.
\end{lemma}

The result above is deterministic in that it does not rely on any
statistical assumptions or arguments.  The next result adds such considerations to give a mean squared error bound for the natural lasso.

\begin{theorem}\label{thm:msebound}
Suppose that each column $\X_j$ of the matrix $\X \in \real^{n \times p}$ has been scaled so that $\snorm{\X_j}^2_2 = n$ for all $j = 1, \ldots, p$, and assume that
  $\varepsilon \sim N \left( \0, \sigma^2 I_n \right)$.
  Then, for any constant $M > 1$,
  the natural lasso estimator \eqref{est:sig_nlasso} with
  $\lambda = \sigma (2Mn^{-1}\log p)^{1/2}$ satisfies the following relative mean squared error bound:
  \begin{align}
    \E \left\{ \left( \frac{\hat{\sigma}^2_\lambda}{\sigma^2} - 1 \right)^2 \right\}
    \leq \left\{\left( 8M + 8\frac{p^{1 - 8 M}}{\log p} \right)^{1/2} \frac{\norm{\beta^\ast}_1}{\sigma}\left( \frac{\log p}{n} \right)^{1/2} + \left( \frac{2}{n} \right)^{1/2} \right\}^2.
    \nonumber
  \end{align}
\end{theorem}
\begin{corollary}
\begin{align}
  \E \left| \frac{\hat{\sigma}^2_\lambda}{\sigma^2} - 1 \right| = O\left\{ \frac{\norm{\beta^\ast}_1}{\sigma}\left( \frac{\log p}{n} \right)^{1/2} \right\}.
  \label{eq:riskbound}
\end{align}
\end{corollary}
\begin{proof}
This follows from Jensen's inequality.
\end{proof}

\begin{remark} \label{rem:polynomialmoment}
  Theorem \ref{thm:msebound} can be easily generalized to the case where the independently and identically distributed zero-mean error $\varepsilon_i$ with variance $\sigma^2$ is sub-Gaussian or sub-exponential. A high probability bound can be obtained for $\varepsilon_i$ with bounded polynomial moments. In particular, for any $m \geq 3$, if $\E(|\varepsilon_i|^m) \leq (m!)^{-1}2 K^{m - 2}$ for some $K > 0$, and if each column $X_j$ is scaled so that $\sum_{i = 1}^n X_{ij}^m = n$ for $j = 1, \dots, p$, then with $\lambda = 4K \sigma n^{-1/2}(\log p)^{1/2}$ we have that
  \begin{align}
    \left| \hat{\sigma}^2_\lambda - \frac{\norm{\varepsilon}_2^2}{n} \right| = O\left\{ \sigma\norm{\beta^\ast}_1\left( \frac{\log p}{n} \right)^{1/2} \right\}    
    \nonumber
  \end{align}
  holds with probability greater than $1 - p^{-1}$.
\end{remark}

To put Theorem \ref{thm:msebound} in context, we devote the remainder of this section to considering what bounds are available for other methods for estimating $\sigma^2$. \citet{bayati2013estimating} propose an estimator of $\sigma^2$ based
on estimating the mean squared error of the lasso.
They show that their estimator of $\sigma^2$ is asymptotically consistent with fixed $p$ as $n \to \infty$. 
In contrast, we provide finite sample results and these include the
$p\gg n$ case. Also, the consistency result in \citet{bayati2013estimating}
is based on the assumption of independent Gaussian features, and
in extending this to the case of correlated Gaussian features, the
authors invoke a conjecture.
In comparison, \eqref{eq:riskbound} is essentially free of assumptions on the design matrix.

The natural lasso also compares favorably to the
method-of-moments-based estimator of \citet{Dicker14MM} in terms of mean squared error bounds. In particular, \citet{Dicker14MM} establishes a 
$O_P[(\sigma^{-2} \tau^2 + 1)\{n^{-2}(p + n)\}^{1/2}]$ relative mean squared error rate, where
$\tau^2 = \snorm{\Sigma^{-1/2}\beta^\ast}_2^2$ and $\Sigma$ is the covariance of features $X$. This rate can be much slower for large $p$.

Notably, the mean squared error bound in Theorem \ref{thm:msebound} does not put any assumption on $\X$, $\beta^\ast$, or $\sigma^2$.  In this sense, the result is analogous to a slow rate bound \citep{rigollet2011exponential,dalalyan2017}, which appears in the lasso prediction consistency context.
While it is well known \citep{Sun25092012} or can be easily verified that under stronger conditions, i.e., compatibility or restricted eigenvalue conditions, the naive estimator \eqref{est:naive} based on the lasso and $\tilde{\sigma}^2_{\mathrm{SQRT}}$ in \eqref{est:sqrt} attain a faster rate, $O( |S| n^{-1} \log p)$,
it is natural to ask whether these two estimators also attain a rate bound as in \eqref{eq:riskbound} when the conditions on $\X$ are not assumed. 
The following two results give an affirmative answer to this question.
\begin{proposition} \label{prop:slow_rate_naive}
  Under the conditions of Theorem \ref{thm:msebound}, the naive estimator \eqref{est:naive} based on the lasso estimator $\hat{\beta}_\lambda$ with $\lambda = 4 \sigma (n^{-1}\log p)^{1/2}$ has the following bound with probability greater than $1 - p^{-1}$:
  \begin{align}
    \left| \snaive - \frac{\snorm{\varepsilon}_2^2}{n} \right| \leq 16 \sigma \snorm{\beta^\ast}_1\left( \frac{\log p}{n} \right)^{1/2} .
    \label{eq:slow_rate_naive}
  \end{align}
\end{proposition}

Relatedly, \citet{chatterjee2015prediction} also consider a setting with no assumptions on $X$ and derive 
an error bound $O\{\snorm{\beta^\ast}_1^{1/2}(n^{-1} \log p)^{1/4}\}$ for \eqref{est:naive} for a lasso estimator $\hat{\beta}_\lambda$ with $\lambda$ in \eqref{est:lasso} selected using a cross-validation procedure.

\citet{2016arXiv160800624L} derive a slow rate bound for the prediction error of the square root lasso.  
They show, in Lemma 2.1, that there exists a value of $\lambda$ for which $\lambda = 3 n^{-1/2}\snorm{\X^T \varepsilon}_\infty \snorm{\y - \X \tilde{\beta}_{\mathrm{SQRT}}}_2^{-1}$ and bound $\snorm{X\tilde\beta_\mathrm{SQRT} - X\beta^\ast}^2_2$ at this value.  The following result establishes the high-dimensional consistency of $\tilde\sigma_\mathrm{SQRT}^2$ under no assumptions on $X$.
\begin{proposition} \label{prop:slow_rate_sqrt}
  Under the conditions of Theorem \ref{thm:msebound}, the square\verb+-+root/scaled lasso estimator $\tilde{\sigma}^2_{\mathrm{SQRT}}$ in \eqref{est:sqrt} based on $\tilde{\beta}_\mathrm{SQRT}$ with 
  $\lambda = 3 n^{-1/2}\snorm{\X^T \varepsilon}_\infty \snorm{\y - \X \tilde{\beta}_{\mathrm{SQRT}}}_2^{-1}$
  has the following bound with probability greater than $1 - p^{-1}$:
  \begin{align}
    \left| \tilde{\sigma}^2_{\mathrm{SQRT}} - \frac{\snorm{\varepsilon}_2^2}{n}  \right|   \leq 12 \sigma \snorm{\beta^\ast}_1\left( \frac{\log p}{n} \right)^{1/2} .
    \label{eq:slow_rate_sqrt}
  \end{align}
\end{proposition}

We see the rate of the natural lasso in \eqref{eq:riskbound} matches, up to a constant factor, the rates \eqref{eq:slow_rate_naive} and \eqref{eq:slow_rate_sqrt}. 
The values of $\lambda$ used in Propositions
\ref{prop:slow_rate_naive} and \ref{prop:slow_rate_sqrt} are larger
than would be necessary for standard prediction error bounds; we
learned of this technique from Irina Gaynanova \citep{irina18commu},
and it is key to the proofs of the two propositions.
Although the same rate is obtained in Theorem \ref{thm:msebound}, Proposition
\ref{prop:slow_rate_naive}, and Proposition \ref{prop:slow_rate_sqrt},
we have not established that this is the best possible rate obtainable in this
setting that makes no assumption on $\X$.  

\section{The organic lasso estimate of error variance} \label{sec:olasso}
\subsection{Method formulation}
In practice, the value of the regularization parameter $\lambda$ in
\eqref{est:sig_nlasso} may be chosen
via cross\verb+-+validation; however, Theorem \ref{thm:msebound} has a regrettable theoretical shortcoming: it requires using a value of $\lambda$
that itself depends on $\sigma$, the very quantity that we are trying
to estimate! This is a well\verb+-+known theoretical limitation of the lasso
and related methods that motivated the square\verb+-+root/scaled lasso. In this section, we propose
a second new method, which retains the natural parameterization, but remedies the natural lasso's theoretical shortcoming by using a modified penalty.  
We define the organic lasso as a solution to \eqref{est:naturalest} with $\Omega(\theta, \phi) =  \phi^{-1}\snorm{\theta}_1^2$, i.e.,
\begin{align}
  \left( \check{\theta}_\lambda, \check{\phi}_\lambda \right)
  = \argmin_{\phi > 0,~\theta} \left( -\frac{1}{2} \log \phi + \phi \frac{\norm{y}_2^2}{2n} - \frac{1}{n} \y^T \X \theta + \frac{\norm{\X\theta}_2^2}{2n\phi} + \lambda \frac{\norm{\theta}_1^2}{\phi} \right).
  \label{est:organiclasso}
\end{align}
We observe that the penalty $\phi^{-1}\snorm{\theta}_1^2$ is jointly
convex in $(\phi,\theta)$ since it can be expressed as $g(h(\theta),
\phi)$ where $h(\theta)=\snorm{\theta}_1$ is convex and
$g(x,\phi)=\phi^{-1}x^2$ is a jointly convex function that is strictly
increasing in $x$ for $x\ge0$ \citep{boyd2004convex, rockafellar2015convex}.

Given a solution to the above problem, we can reverse \eqref{est:naturalpara} to give the organic lasso estimators of $(\beta^\ast,\sigma^2)$, i.e.,
$\check{\beta}_\lambda = \check{\phi}_\lambda^{-1} \check{\theta}_\lambda, \check{\sigma}^2_\lambda = \check{\phi}_\lambda^{-1}$.
Furthermore, $\phi^{-1}\snorm{\theta}_1^2$ still induces sparsity in $\theta$, and thus the final estimate $\check{\beta}_\lambda$ is sparse.
In direct analogy to the natural lasso, the following proposition shows that we can find $\check{\sigma}_\lambda^2$ and $\check{\beta}_\lambda$ without actually solving \eqref{est:organiclasso}.
\begin{proposition} \label{main:prop:ol}
The organic lasso estimators $(\check\beta_\lambda,\check\sigma^2_\lambda)$ correspond to the solution and minimal value of an $\ell_1^2$\verb+-+penalized least\verb+-+squares problem:
\begin{align}
  \check{\beta}_\lambda = \argmin_{\beta \in \real^p} \left( \frac{1}{n} \norm{\y - \X \beta}_2^2 + 2 \lambda \norm{\beta}_1^2 \right);
  \label{est:beta_olasso}
  \\
  \check{\sigma}^2_\lambda = \min_{\beta \in \real^p} \left( \frac{1}{n} \norm{\y - \X \beta}_2^2 + 2 \lambda \norm{\beta}_1^2 \right).
  \label{est:sig_olasso}
\end{align}
\end{proposition}
Thus, to compute the organic lasso estimator, one simply solves a
penalized least squares problem, where the penalty is the square of
the $\ell_1$ norm. This can be thought of as the exclusive lasso with a single group \citep{zhou2010exclusive, campbell2017within}. We show in the next section that solving this problem is no harder than solving a standard lasso problem.

One readily sees the connection of the organic lasso to the square\verb+-+root lasso \eqref{est:sqrt}: to get \eqref{est:sig_olasso}, one takes squares of both the loss and the $\ell_1$ penalty of \eqref{est:sqrt}. However, their origins are actually different in nature:
the organic lasso is a maximum of the Gaussian log\verb+-+likelihood with a scale\verb+-+equivariant sparsity inducing penalty under parameterization \eqref{est:naturalpara}, while \eqref{est:sqrt} minimizes the $\ell_1$\verb+-+penalized Huber concomitant loss function \citep{antoniadis2010comments, Sun25092012}.

\subsection{Algorithm}
Coordinate descent is easy to implement and has steadily maintained its place as a start\verb+-+of\verb+-+the\verb+-+art approach for solving lasso\verb+-+related problems \citep{friedman2007pathwise}.
For coordinate descent to work, one typically verifies separability in the non\verb+-+smooth part of the objective function \citep{Tseng2001}.
However, the $\ell_1^2$ penalty in \eqref{est:beta_olasso} is not separable in the coordinates of $\beta$.
\citet{lorbert2010exploiting} propose a coordinate descent algorithm to solve the Pairwise Elastic Net (PEN) problem, a generalization of
\eqref{est:beta_olasso}, and a proof of the convergence of the algorithm is given in \citet{lorbert:2012aa}.
In Algorithm \ref{alg:cd}, we give a coordinate descent algorithm specific to solving \eqref{est:beta_olasso}.
The \texttt{R} package \texttt{natural} \citep{citenatural} provides a \texttt{C} implementation of Algorithm \ref{alg:cd}.

\begin{algorithm}
  \caption{A coordinate descent algorithm to solve \eqref{est:beta_olasso}}
  \label{alg:cd}
\begin{algorithmic}
   \Require Initial estimate $\beta^{(0)} \in \real^p$, $\X \in \real^{n \times p}$, $\y \in \real^n$, and $\lambda > 0$.
   \State Set $\beta \gets \beta^{(0)}$ and $\r \gets \y - \X \beta$
   \For{$j=1,\ldots,p;1,\ldots,p;\ldots$ (until convergence)}
   \State $\beta_j^{\mathrm{new}} \gets (2 \lambda + n^{-1} \norm{\X_j}_2^2)^{-1} \soft ( n^{-1} \X_j^T \r + n^{-1} \norm{\X_j}_2^2 \beta_j, 2 \lambda \snorm{\beta_{-j}}_1 )$
   \State $\r \gets \r + \X_j \beta_j - \X_j \beta_j^{\mathrm{new}}$
   \State $\beta_j \gets \beta_j^{\mathrm{new}}$
   \EndFor
  \State \Return $\beta$.
\end{algorithmic}
\end{algorithm}

Each coordinate update requires $O(n)$ operations. In Algorithm \ref{alg:cd}, $\soft(a, b)=\text{sgn}(a)(|a|-b)_+$ is the soft\verb+-+threshold operator.
Empirically Algorithm \ref{alg:cd} is as fast as solving a lasso problem.
Theorem C.3.9 in \citet{lorbert:2012aa} shows that,
for any initial estimate $\beta^{(0)} \in \real^p$, every limit point of Algorithm \ref{alg:cd} is an optimal point of the objective function of \eqref{est:beta_olasso}.
This implies that the $\ell_1^2$ penalty, although not separable, is well enough behaved that any point that is minimum in every coordinate of the objective function in \eqref{est:beta_olasso} is indeed a global minimum.

\subsection{Theoretical results}
A first indication that the organic lasso may succeed where the natural lasso falls short is in terms of scale equivariance.
As the design $\X$ is usually standardized to be unitless, scale
equivariance in this context refers to the effect of scaling $\y$.
\begin{proposition}
 \label{thm:scaleequivariant}
  The organic lasso is scale equivariant, i.e., for any $t > 0$,
  \begin{align}
    \check{\beta}_\lambda \left( t\y\right) = t \check{\beta}_\lambda\left( \y \right),
    \qquad
    \check{\sigma}_\lambda \left( t\y \right) = t \check{\sigma}_\lambda \left( \y \right).
    \nonumber
  \end{align}
\end{proposition}
Scale equivariance is a property associated with the ability to prove results in which the tuning parameter $\lambda$ does not depend on $\sigma$.
For example, the square\verb+-+root/scaled lasso \eqref{est:sqrt} is scale equivariant while the lasso, and thus the natural lasso, is not.
In particular, $\hat{\beta}_\lambda(t\y) \neq t \hat{\beta}_\lambda (\y)$, and $\hat{\sigma}_\lambda (t\y) \neq t \hat{\sigma}_\lambda(\y)$ for some $t > 0$.

In Lemma \ref{lem1}, we saw how expressing an estimator as the optimal value
of a convex optimization problem allows us to take full advantage of convex duality in order to derive bounds on the estimator.
We therefore start our analysis of \eqref{est:sig_olasso} by characterizing its dual problem.
\begin{lemma} \label{lem:dual}
  The dual problem of \eqref{est:sig_olasso} is 
  \begin{align}
    \max_{u \in \real^n} \left\{ \frac{1}{n} \left( \norm{\y}^2_2 - \norm{\y - u}^2_2 \right) - \frac{1}{2\lambda} \norm{\frac{\X^T u}{n}}_\infty^2 \right\}.
    \nonumber
  \end{align}
\end{lemma}
Similar arguments as in Lemma \ref{lem1} give a bound expressing $\check{\sigma}_\lambda^2$'s closeness to the oracle estimator of $\sigma^2$.
\begin{lemma} \label{lem:close_to_oracle_organic}
  If $\lambda \geq n^{-1}\norm{\X^T (\sigma^{-1}\varepsilon)}_\infty$, then
\begin{align}
  - 2 \lambda \sigma^2 \left( \frac{\norm{\beta^\ast}_1}{\sigma} + \frac{1}{4} \right) \leq \check{\sigma}_\lambda^2 - \frac{1}{n} \norm{\varepsilon}^2_2 \leq 
  2 \lambda \norm{\beta^\ast}_1^2.
  \nonumber
\end{align}
\end{lemma}
Comparing with Lemma
\ref{lem1}, we see that the condition on $\lambda$ depends only on
a quantity $\sigma^{-1}\varepsilon\sim N(0,I_n)$ that is independent of $\sigma^2$.
Indeed, this leads to a mean squared error bound with the desired property of $\lambda$ not depending on $\sigma$.
\begin{theorem}\label{thm:msebound_olasso}
  Suppose that each column $\X_j$ of the matrix $\X \in \real^{n \times p}$ has been scaled so that $\snorm{\X_j}^2_2 = n$ for all $j = 1, \ldots, p$, and
  $\varepsilon \sim N \left( \0, \sigma^2 I_n \right)$.
  Then, for any constant $M > 1$,
  the organic lasso estimator \eqref{est:sig_olasso} with
  $\lambda = (2M n^{-1}\log p)^{1/2}$ satisfies the following relative mean squared error bound:
  \begin{align}
    \E \left\{ \left( \frac{\check{\sigma}^2_\lambda}{\sigma^2} - 1 \right)^2 \right\} 
    \leq \left\{\left( 8M + 8\frac{p^{1 - 8 M}}{\log p} \right)^{1/2}\max\left( \frac{\norm{\beta^\ast}_1^2}{\sigma^2}, \frac{\norm{\beta^\ast}_1}{\sigma} + \frac{1}{4} \right)\left( \frac{\log p}{n} \right)^{1/2} + \left( \frac{2}{n} \right)^{1/2} \right\}^2.
    \label{eq:riskbound_ratio_olasso}
  \end{align}
\end{theorem}
Compared with Theorem \ref{thm:msebound}, the organic lasso estimator
of $\sigma^2$ retains the same rate in terms of $n$ and $p$ but has a slower rate in terms of 
$\sigma^{-1} \norm{\beta^\ast}_1$.
Importantly, though, the value of $\lambda$ attaining \eqref{eq:riskbound_ratio_olasso} does not depend on $\sigma$. This tuning\verb+-+insensitive property is also enjoyed by the square-root/scaled lasso estimate of $\sigma^2$, as shown in Proposition \ref{prop:slow_rate_sqrt}. As in Remark \ref{rem:polynomialmoment}, similar high\verb+-+probability bounds can be obtained for $\varepsilon$ with bounded polynomial moments.

Although not central to our main purpose, the organic lasso estimator
\eqref{est:beta_olasso} of $\beta^\ast$ is interesting in its own right. The following theorem
gives a slow rate bound in prediction error.
\begin{theorem} \label{thm:olasso_predict}
  For any $L > 0$, the solution to \eqref{est:beta_olasso} with
  $\lambda = \{2n^{-1}(\log p + L)\}^{1/2}$
  has the following bound on the prediction error with probability greater than $1 - e^{-L}$:
  \begin{align}
    \frac{1}{n} \norm{\X \check{\beta}_\lambda - \X \beta^\ast}_2^2 \leq
    \left( \sigma^2 + 4 \norm{\beta^\ast}_1^2 \right)\left( \frac{2 \log p + 2L}{n} \right)^{1/2}.
    \nonumber
  \end{align}
\end{theorem}

In Appendix \ref{proof:equivalence}, we provide mappings between the path
of the natural lasso, $\{\hat\beta_\lambda:\lambda>0\}$, and the path
of the organic lasso $\{\check\beta_\lambda:\lambda>0\}$. We also include a fast\verb+-+rate prediction error bound of \eqref{est:beta_olasso} under a compatibility condition in Appendix \ref{proof:fastrate}.

\section{Simulation studies} \label{sec:simulation}
\subsection{Simulation settings}
\citet{reid2013study} carry out an extensive simulation study to
compare many error variance estimators. 
We have matched their simulation settings fairly closely, so that the performance comparison with various other methods mentioned in \citet{reid2013study}
can be inferred.
Specifically, all simulations are run with $p = 500$ and $n = 100$. 
Each row of the design $\X$ is generated from a multivariate $N( \0, \Sigma)$, with $\Sigma_{ij} = \rho \in (0, 1)$ for $i\neq j$ and $\Sigma_{ii} = 1$.
To generate $\beta^\ast$, we randomly select the indices of $\lceil n^\alpha \rceil$ nonzero elements out of $p$ variables where $\alpha \in (0, 1)$,
and each of the nonzero elements has a value that is randomly drawn from a Laplace distribution with rate 1.
The error variance is generated using $\sigma^2 = \tau^{-1} {\beta^\ast}^T\Sigma \beta^\ast$ for $\tau > 0$.
Finally, $\y$ is generated following \eqref{est:model}.

Each model is indexed by a triplet $(\rho, \alpha, \tau)$, where $\rho$ captures the correlation among features, $\alpha$ determines the sparsity of $\beta^\ast$,
and $\tau$ characterizes the signal\verb+-+to\verb+-+noise ratio. We vary $\rho,\alpha\in\{0.1, 0.3, 0.5, 0.7, 0.9\}$ and $\tau\in\{0.3, 1, 3\}$. 
We compute a Monte Carlo estimate of both the mean squared error $E\{(\sigma^{-1}\hat{\sigma} - 1)^2\}$ and $E(\sigma^{-1} \hat{\sigma} )$ as the measure of performance.
The methods in comparison include
(a) the naive estimator \eqref{est:naive} with $\hat{\beta}_\lambda$ in \eqref{est:lasso}, 
(b) the degrees of freedom adjusted estimator $\rt$ in \eqref{est:reid} \citep{reid2013study},
(c) the square\verb+-+root/scaled lasso \citep{belloni2011square, sun2013sparse},
(d) the natural lasso \eqref{est:sig_nlasso},
and (e) the organic lasso \eqref{est:sig_olasso}.
As a benchmark, we also include the oracle $n^{-1} \snorm{\varepsilon}_2^2$.
The \texttt{simulator} R package \citep{citesimulator} was used for all simulations.

\subsection{Methods with regularization parameter selected by cross-validation}
We carry out two sets of simulations.
In the first set, we compare the performance of the aforementioned methods with regularization parameter selected in a data\verb+-+adaptive way.
In particular, five\verb+-+fold cross\verb+-+validation is used to select the tuning parameter for each method.
\begin{figure}
  \centering
  \includegraphics[width = .8\textwidth]{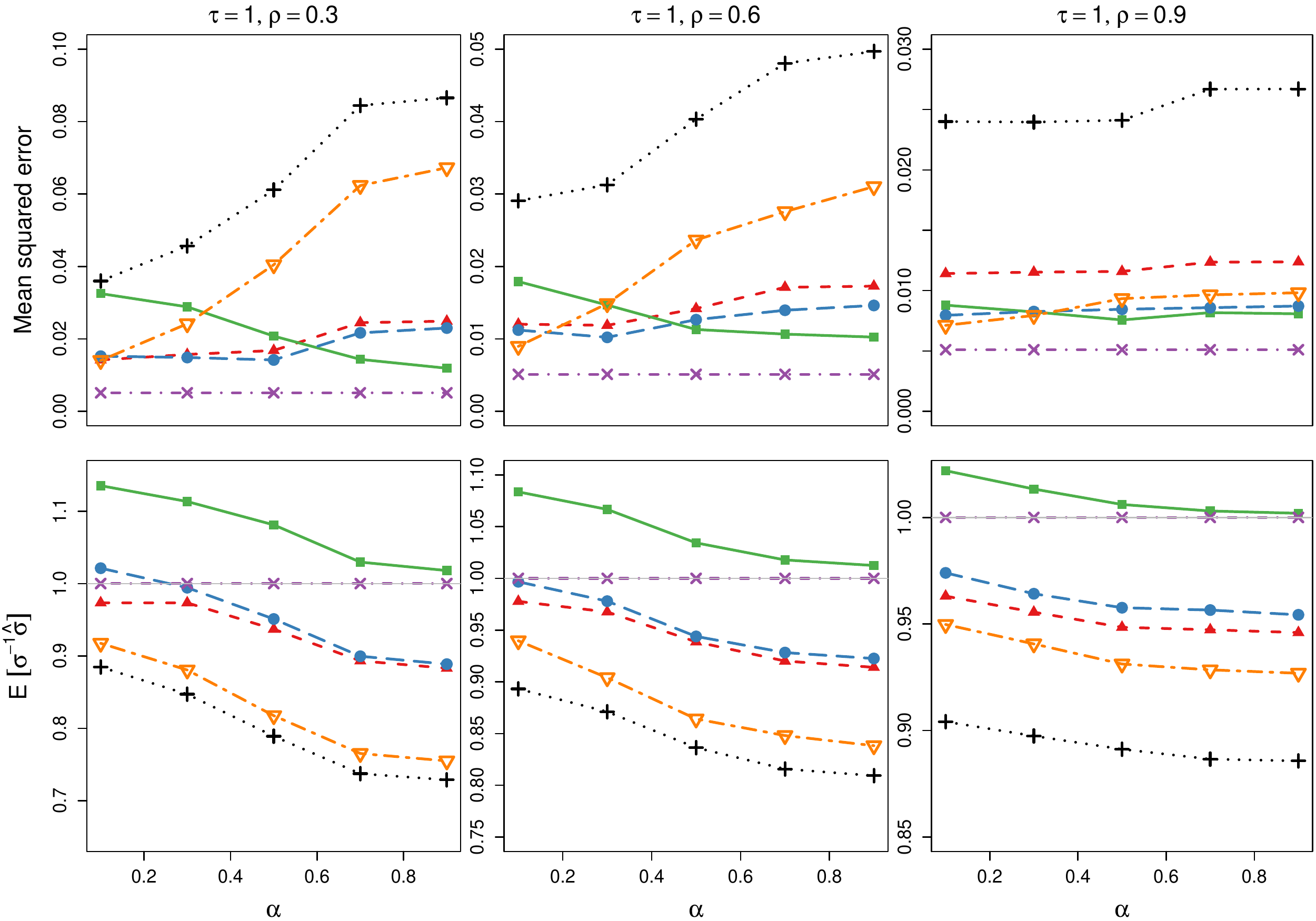}
  \caption{Simulation results of methods using cross-validation. From left to right, columns show Monte Carlo estimates of
   the mean squared error, in the top panel, and $E(\sigma^{-1}\hat{\sigma})$, in the bottom panel, of various methods in three simulation
   settings. Line styles and their corresponding methods: 
    \protect\includegraphics[width = 0.12in]{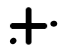} for naive, 
    \protect\includegraphics[width = 0.12in]{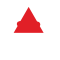} for $\rt$, 
    \protect\includegraphics[width = 0.12in]{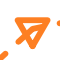} for the square-root/scaled lasso,
    \protect\includegraphics[width = 0.12in]{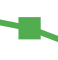} for the natural lasso, 
    \newline
    \protect\includegraphics[width = 0.12in]{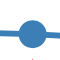} for the organic lasso, 
    \protect\includegraphics[width = 0.12in]{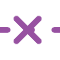} for the oracle.}
    \label{fig:cv}
  \end{figure}

Due to space constraints, we present a subset of the results in
Fig~\ref{fig:cv}. Additional results are presented in 
Appendix \ref{app:numerical}.
The result for the square\verb+-+root/scaled lasso is averaged over 100 repetitions due to the large computational time.
For all other methods, the results are averaged over 1000 repetitions.
Overall, the natural lasso does well in adjusting the downward bias of the naive estimator,
while other methods tend to produce under\verb+-+estimates.
In each panel, we fix signal\verb+-+to\verb+-+noise ratio ($\tau$) and correlations
among features ($\rho$), and vary model sparsity ($\alpha$).
All estimates get worse with growing $\alpha$, except for the natural lasso, which improves as the true $\beta^\ast$ gets denser.
In particular, both the natural lasso and the organic lasso gain performance advantage over other methods when the underlying models
do not satisfy conditions for the support recovery of the lasso solution.
From left to right, Fig~\ref{fig:cv} illustrates the effect of increasing $\rho$.
As observed in \citet{reid2013study}, high correlations can be
helpful: All curves approach the oracle as $\rho$ increases.
Finally, we find that the organic lasso is uniformly better or equivalent to $\rt$.

Paired $t$\verb+-+tests and Wilcoxon signed\verb+-+rank tests show that the
differences in mean squared errors of different methods are
significant at the $5\%$ level for almost
all points shown in Fig~\ref{fig:cv}.

Results in Appendix \ref{app:numerical} also show the natural lasso
estimator doing well when the signal\verb+-+to\verb+-+noise ratio is low:
the performances of all methods degrade as $\tau$ gets large.
This is expected from Theorem \ref{thm:msebound} and Theorem \ref{thm:msebound_olasso}, and is also observed in \citet{reid2013study}.

\subsection{Methods with fixed choice of regularization parameter}
Although solving \eqref{est:sig_olasso} is fast enough for one to
use cross\verb+-+validation with the organic lasso,
Theorem \ref{thm:msebound_olasso} implies that $\lambda_0 = ({2n^{-1}\log p})^{1/2}$ is a theoretically sound choice of regularization parameter.
We also conjecture that a sharper rate may be obtainable at $\lambda_1 \geq
n^{-2}\snorm{\X^T \epsilon}_\infty^2 $, where $\epsilon \sim N(0,
1)$. With high probability, $n^{-2}\snorm{\X^T \epsilon}_\infty^2 \approx n^{-1}\log(p)$.
Thus, we also show the performance of the organic lasso with tuning
parameter values equal to $\lambda_2 = n^{-1} \log(p)$,
and $\lambda_3$, which is a Monte Carlo estimate of $E(n^{-2}\snorm{\X^T\epsilon}_\infty^2)$, where the expectation is with respect to
$\epsilon \sim N(0, 1)$.

We compare the organic lasso at these three fixed values of tuning
parameter to the \newline square\verb+-+root/scaled lasso estimator \eqref{est:sqrt} of error variance,
which is another method
whose theoretical choice of $\lambda$ does not depend on $\sigma$.
\citet{Sun25092012} find that $\lambda_0$ works very well
for \eqref{est:sqrt}, which we denote by scaled(1), and
\citet{sun2013sparse} propose a refined choice of $\lambda$, which is
proved to attain a sharper rate, denoted by scaled(2).  
The results of all the methods are averaged over 1000 repetitions.

\begin{figure}
  \centering
 \includegraphics[width = .8\textwidth]{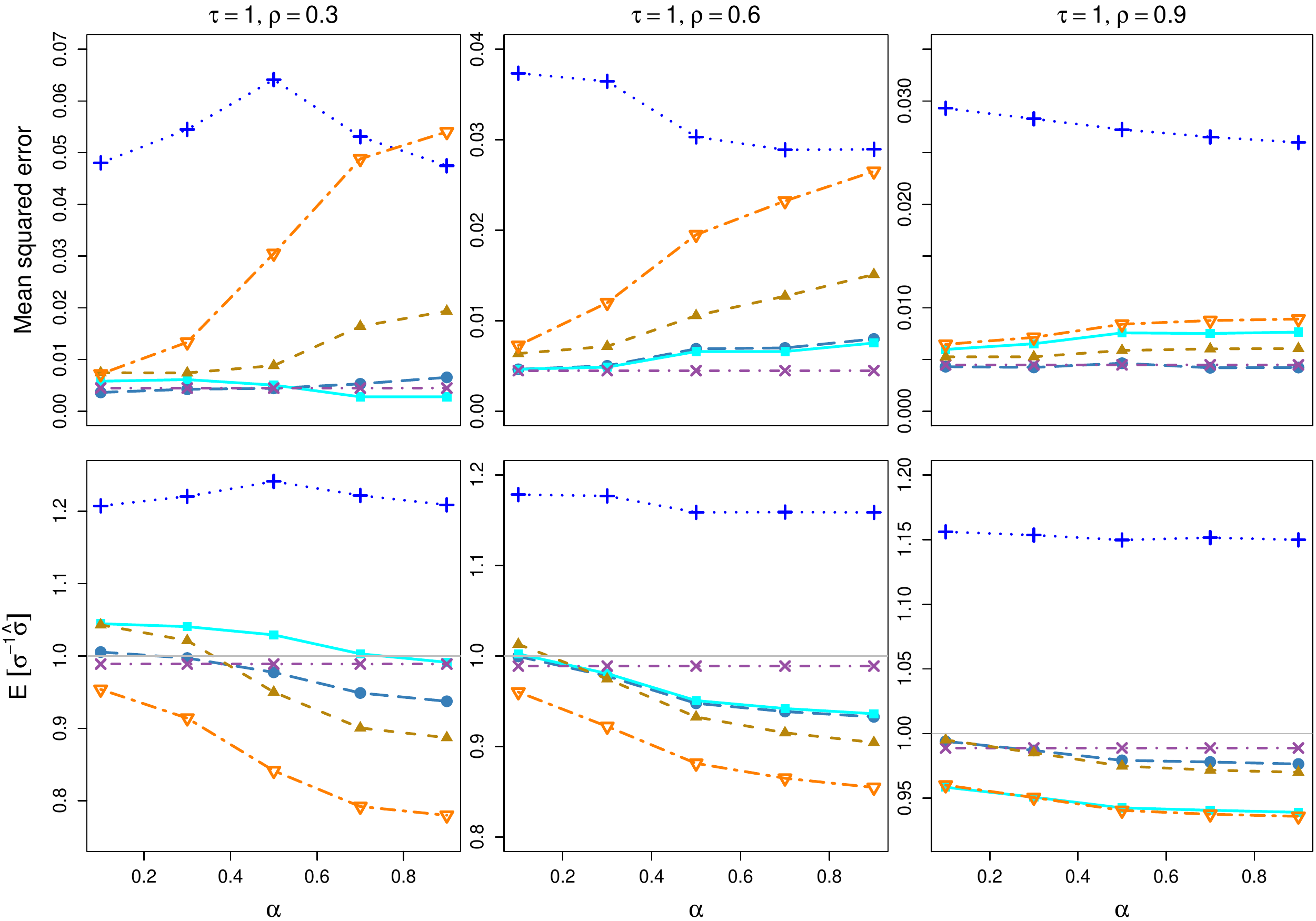}
  \caption{Simulation results of methods using pre-specified
    regularization parameter values. 
    From left to right, columns show
    Monte Carlo estimates of the mean squared error, in the top panel,
    and $E(\sigma^{-1} \hat{\sigma})$, in the bottom panel, of various methods in three simulation settings. 
    Line styles and their
    corresponding methods:
    \protect\includegraphics[width = 0.12in]{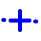} for organic ($\lambda_0$),
    \protect\includegraphics[width = 0.12in]{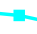} for organic ($\lambda_2$), 
    \protect\includegraphics[width = 0.12in]{blue.png} for organic ($\lambda_3$), 
    \protect\includegraphics[width = 0.12in]{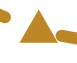} for scaled(1),
    \protect\includegraphics[width = 0.12in]{orange.png} for scaled (2), 
  \protect\includegraphics[width = 0.12in]{purple.png} for the oracle.}
  \label{fig:fix}
\end{figure}

Fig~\ref{fig:fix} shows similar patterns as
Fig~\ref{fig:cv}. Specifically, large value of $\rho$ helps all methods, 
while performance generally degrades for denser $\beta^\ast$.
Although not shown here, all methods struggle as $\tau$ increases.
The theoretically justified tuning parameter $\lambda_0$ for the
organic lasso appears in practice to overshrink the estimate of
$\beta^*$ and thus to overestimate $\sigma^2$, leading to poor performance; however, the organic
lasso with the smaller tuning parameter values $\lambda_2$ and $\lambda_3$ do quite well, generally outperforming
 the square\verb+-+root/scaled lasso based methods.

\section{Error estimation for Million Song dataset} \label{sec:realdata}
We apply our error variance estimators to the Million Song dataset.\footnote{The whole data set can be obtained at https://labrosa.ee.columbia.edu/millionsong/. We consider a subset of the whole data, which is available at
https://archive.ics.uci.edu/ml/datasets/yearpredictionmsd.} The data
consist of information about 463715 songs, and the primary goal is to
model the release year of a song using $p = 90$ of its timbre features.
The dataset has a very large sample size so that we can reliably estimate
the ground truth of the target of estimation on a very large set of held out data.
In particular, we randomly select half of the songs for this purpose and use 
$\bar{\sigma}^2 = (n - p)^{-1}\snorm{\y - \X \hat{\beta}_{LS}}_2^2$ to form our ground
truth, where $\hat{\beta}_{LS}$ is the least\verb+-+squares estimator of
$\beta^\ast$.
In practice, model \eqref{est:model} may rarely hold, which alters the
interpretation of error variance estimation. Suppose the response vector
$y$ has mean $\mu$ and covariance matrix $\Sigma$.  Then
$\bar\sigma^2$ can be thought of as an estimator of the population quantity  
\begin{align}
\min_{\beta} \frac{1}{n} E\left(\norm{\y - \X \beta}_2^2\right) = \frac{1}{n} \tr (\Sigma) + \frac{1}{n} \norm{\left( I - XX^+ \right)\mu}_2^2. \nonumber
\end{align}
In the special case where $\Sigma=\sigma^2 I_n$ and $\mu=X\beta^\ast$,
as in \eqref{est:model}, then $\bar{\sigma}^2$ reduces to the linear model noise variance $\sigma^2$.

From the remaining data that was not previously used to yield $\bar{\sigma}^2$, we randomly
form training datasets of size $n$ and compare the performance of
various error variance estimators.  
We vary $n$ in $\{20, 40, 60, 80, 100, 120\}$ to gauge the performance of these methods in situations in which $n < p$ and $n
\approx p$.
For each $n$, we repeat the data selection and error variance
estimation on $1000$ disjoint training sets, and report estimates of the mean squared error 
$E\{(\bar\sigma^{-1} \hat{\sigma} - 1)^2\}$ in Table~\ref{tab:MSD}
and estimates of $E (\bar{\sigma}^{-1}\hat{\sigma} )$ in Appendix \ref{app:numerical}.

\begin{table}[ht]
  \def~{\hphantom{0}}
  \caption{Mean squared error of noise variance estimation for Million Song dataset}
    \centering
    \begin{tabular}{lcccccc}
      \hline
      n & 20 & 40 & 60 & 80 & 100 & 120 \\ 
      \hline
      naive & 17.02 (0.68) & 8.48 (0.41) & 5.28 (0.26) & 3.80 (0.17) & 3.03 (0.13) & 2.43 (0.10) \\ 
      $\rt$ & 10.74 (0.45) & 5.92 (0.29) & 3.57 (0.17) & 2.57 (0.11) & 2.23 (0.10) & 1.75 (0.08) \\ 
      natural(cv) & ~8.82 (0.38) & 5.23 (0.27) & 3.47 (0.16) & 2.61 (0.12) & 2.39 (0.11) & 2.01 (0.09) \\ 
      organic(cv) & ~8.08 (0.32) & 4.23 (0.20) & 2.59 (0.12) & 2.00 (0.08) & 1.72 (0.08) & 1.54 (0.07) \\ 
      scaled(1) & ~7.43 (0.37) & 4.92 (0.25) & 3.84 (0.17) & 3.08 (0.13) & 2.94 (0.12) & 2.75 (0.11) \\ 
      scaled(2) & ~7.11 (0.28) & 3.36 (0.15) & 2.23 (0.10) & 2.57 (0.83) & 1.61 (0.07) & 1.46 (0.07) \\ 
      organic($\lambda_2$) & ~5.87 (0.24) & 3.17 (0.14) & 1.93 (0.09) & 1.40 (0.06) & 1.20 (0.05) & 1.02 (0.05) \\ 
      organic($\lambda_3$) & ~5.72 (0.24) & 3.15 (0.14) & 1.99 (0.09) & 1.45 (0.07) & 1.28 (0.05) & 1.12 (0.05) \\
      \hline
    \end{tabular}
  \\[10pt]
  \caption*{Mean and standard errors, over 1000 replications, of the squared error of various methods. Each entry is multiplied by 100 to convey information more compactly.}
  \label{tab:MSD}
\end{table}

All methods produce a substantial performance improvement 
over the naive estimator for a wide range of values of $n$.
The natural and organic lassos with cross validation perform either
better or comparably to $\rt$ and are in some, but not all, cases
outperformed by scaled(2).
When $n$ gets large, the natural lasso shows some upward bias, which as we noted before is less problematic than downward bias.
The organic lasso with the fixed choices $\lambda_2$ or $\lambda_3$ performs extremely well for all $n$.

Future research directions include the analysis of the proposed methods with smaller values of $\lambda$, and extending the natural parameterization to penalized non\verb+-+parametric regression. Finally, an \texttt{R} \citep{citeR} package, named \texttt{natural} \citep{citenatural}, is available on the Comprehensive R Archive Network, implementing our estimators.

\section*{Acknowledgement}
We thank Irina Gaynanova for a useful conversation that helped us prove Propositions \ref{prop:slow_rate_naive} and \ref{prop:slow_rate_sqrt}.
JB was supported by an NSF CAREER grant, DMS-1653017.

\newpage
\begin{appendices}
\section{Proof of Lemma \ref{lem1}}
From \eqref{est:lasso} in the paper, it follows that
\begin{align}
  \hat{\sigma}^2_\lambda \leq \frac{1}{n} \norm{\y - \X\beta^\ast}^2_2 + 2\lambda \norm{\beta^\ast}_1 = \frac{1}{n} \norm{\varepsilon}^2_2 + 2\lambda \norm{\beta^\ast}_1.
  \nonumber
\end{align}
By introducing the dual variable $2n^{-1}u \in \real^n$, 
\begin{align}
  \hat{\sigma}^2_\lambda =& \min_{\beta} \left( \frac{1}{n}\norm{\y - \X\beta}^2_2 + 2\lambda \norm{\beta}_1 \right)
  = \min_{\beta, \z} \max_{u} \left\{ \frac{1}{n} \norm{\y - \z}^2_2 + \frac{2}{n}u^T \left( \z - \X \beta \right) + 2\lambda \norm{\beta}_1 \right\}
  \nonumber\\
  \ge&\max_{u} \min_{\beta, \z} \left\{ \frac{1}{n} \norm{\y - \z}^2_2 + \frac{2}{n}u^T \left( \z - \X \beta \right) + 2\lambda \norm{\beta}_1 \right\}
  \nonumber\\
  =& \max_{u} \left( \frac{1}{n}\norm{\y}^2_2 - \frac{1}{n}\norm{\y - u}^2_2, \text{subject to} \,\, \norm{\X^Tu}_\infty \leq n \lambda \right).
  \nonumber
\end{align}
By assumption, $\varepsilon$ is dual feasible, which means that
\begin{align}
  \hat{\sigma}^2_\lambda
  \geq \frac{1}{n} \norm{\y}^2_2 - \frac{1}{n} \norm{\y - \varepsilon}^2_2 
  \geq  \frac{1}{n} \norm{\varepsilon}^2_2 + \frac{2}{n} \varepsilon^T \X \beta^\ast
  \geq \frac{1}{n} \norm{\varepsilon}^2_2 - 2 \lambda \norm{\beta^\ast}_1,
  \nonumber
\end{align}
where in the last step we applied H{\"o}lder's inequality.

\section{Proof of Propositions \ref{main:prop:nl} and \ref{main:prop:ol}}
\label{sec:pf-prop:nl}
We prove in this section that both the natural lasso and the organic lasso estimates of error variance can be simply expressed as the minimizing values of certain convex optimization problems. To do so, we exploit the first order optimality condition of each convex program. 

We start with proving that the natural lasso estimate of $\sigma^2$ is the minimal value of a lasso problem \eqref{est:lasso}. The following lemma characterizes the conditions for which $(\hat{\theta}_\lambda, \hat{\phi}_\lambda)$ is a solution to \eqref{est:naturalest} with $\Omega(\theta, \phi) = \snorm{\theta}_1$.
\begin{lemma}[Optimality condition of the natural lasso] \label{lem:optimality_naturallasso}
  For any $\lambda > 0$, $( \hat{\theta}_\lambda, \hat{\phi}_\lambda)$ is a solution to \eqref{est:naturalest} with $\Omega(\theta, \phi) = \snorm{\theta}_1$ if and only if 
  \[
    - \frac{1}{\hat{\phi}_\lambda} + \frac{1}{n} \norm{\y}^2_2 - \frac{\norm{\X\hat{\theta}_\lambda}^2_2}{n\hat{\phi}_\lambda^2} = 0,
    \qquad
    -\X^T \y + \X^T \X \frac{\hat{\theta}_\lambda}{\hat{\phi}_\lambda} + n\lambda \hat{g} = \0
  \]
  where $\hat{g} \in \partial ( \snorm{\hat{\theta}_\lambda}_1 )$.
\end{lemma}

Given $( \hat{\theta}_\lambda, \hat{\phi}_\lambda )$, 
we reverse the natural parameterization to get $\hat{\beta}_\lambda = \hat{\phi}_\lambda^{-1}\hat{\theta}_\lambda$ and 
$\hat{\sigma}^2_\lambda = \hat{\phi}_\lambda^{-1}$.
From Lemma \ref{lem:optimality_naturallasso},
\begin{align}
  \hat{\sigma}^2_\lambda &= \frac{1}{n} \left( \norm{\y}^2_2 - \norm{\X \hat{\beta}_\lambda}^2_2 \right) \quad \text{and} \quad
  0 = -\hat{\beta}_\lambda^T \X^T \y + \norm{\X\hat{\beta}_\lambda}^2_2 + n \lambda \norm{\hat{\beta}_\lambda}_1.
  \nonumber
\end{align}
Note that 
\begin{align}
  \norm{\y - \X \hat{\beta}_\lambda}^2_2
  = \norm{\y}^2_2 - \norm{\X \hat{\beta}_\lambda}^2_2 + 2 \left( \norm{\X\hat{\beta}_\lambda}^2_2 - \y^T \X \hat{\beta}_\lambda \right)
  = \norm{\y}^2_2 - \norm{\X \hat{\beta}_\lambda}^2_2 - 2 n \lambda \norm{\hat{\beta}_\lambda}_1.
  \nonumber
\end{align}
We have
\begin{align}
  \hat{\sigma}^2_\lambda = \frac{1}{n} \left( \norm{\y}^2_2 - \norm{\X\hat{\beta}_\lambda}^2_2 \right) = \frac{1}{n} \norm{\y - \X\hat{\beta}_\lambda}^2_2 + 2\lambda \norm{\hat{\beta}_\lambda}_1.
  \nonumber
\end{align}

We show that the organic lasso estimate of $\sigma^2$ is the minimal value of the $\ell_1^2$-penalized least squares problem. As the natural lasso, we start with studying the following optimality condition:
\begin{lemma}[Optimality condition of the organic lasso] \label{lem:optimality_organiclasso}
  For any $\lambda > 0$, $\left(\check{\theta}_\lambda, \check{\phi}_\lambda \right)$ is a solution to \eqref{est:organiclasso} if and only if 
  \begin{align}
    - \frac{1}{\check{\phi}_\lambda} + \frac{1}{n} \norm{\y}^2_2 - \frac{\norm{\X\check{\theta}_\lambda}^2_2}{n\check{\phi}_\lambda^2} - 2\lambda \frac{\norm{\check{\theta}_\lambda}_1^2}{\check{\phi}_\lambda^2}= 0,
    \qquad 
    -\X^T \y + \X^T \X \frac{\check{\theta}_\lambda}{\check{\phi}_\lambda} + 2 n\lambda \frac{\norm{\check{\theta}_\lambda}_1}{\check{\phi}_\lambda} \check{g} = \0
    \nonumber
  \end{align}
  where $\check{g} \in \partial ( \snorm{\check{\theta}}_1)$.
\end{lemma}
So following the natural parameterization, we have that
$\check{\beta}_\lambda = {\check{\theta}_\lambda^{-1}}{\check{\rho}_\lambda}$ and $\check{\sigma}^2 _\lambda= \check{\rho}_\lambda^{-1}$, and
\begin{align}
  \check{\sigma}^2_\lambda &= \frac{1}{n} \left( \norm{\y}^2_2 - \norm{\X \check{\beta}_\lambda}^2_2 - 2 n \lambda \norm{\check{\beta}_\lambda}_1^2 \right) \nonumber\\
  0 &= -\check{\beta}^T_\lambda \X^T \y + \norm{\X\check{\beta}_\lambda}^2_2 + 2 n \lambda \norm{\check{\beta}_\lambda}_1^2.
  \nonumber
\end{align}
Note that 
\begin{align}
  \norm{\y - \X \check{\beta}_\lambda}^2_2 &= \norm{\y}^2_2 + \norm{\X\check{\beta}_\lambda}^2_2 - 2 \y^T \X\check{\beta}_\lambda
  \nonumber\\
  &= \norm{\y}^2_2 - \norm{\X \check{\beta}_\lambda}^2_2 + 2 \left( \norm{\X\check{\beta}_\lambda}^2_2 - \y^T \X \check{\beta}_\lambda \right)
  \nonumber\\
  &= \norm{\y}^2_2 - \norm{\X \check{\beta}_\lambda}^2_2 - 4 n \lambda \norm{\check{\beta}_\lambda}_1^2.
  \nonumber
\end{align}
We have
\begin{align}
  \check{\sigma}^2_\lambda = \frac{1}{n} \left( \norm{\y}^2_2 - \norm{\X\check{\beta}_\lambda}^2_2 - 2 n \lambda \norm{\check{\beta}_\lambda}_1^2 \right) = \frac{1}{n} \norm{\y - \X\check{\beta}_\lambda}^2_2 + 2\lambda \norm{\check{\beta}_\lambda}_1^2.
  \nonumber
\end{align}

\section{Proof of Lemma \ref{lem:dual}: the dual problem of the $\ell_1^2$-penalized least squares}
The primal problem of the $\ell_1^2$-penalized least squares \eqref{est:beta_olasso} in the paper can be written as an equality constrained minimization problem:
\begin{align}
  \min_{\beta \in \real^p} \left( \frac{1}{n} \norm{\y - \z}^2_2 + 2\lambda \norm{\beta}_1^2 \subto \frac{2}{n} \z = \frac{2}{n} \X\beta \right).
  \nonumber
\end{align}
The Lagrange dual function is 
\begin{align}
  g\left( u \right) &= \min_{\beta \in \real^p, \z \in \real^n} \left\{ \frac{1}{n} \norm{\y - \z}^2_2 + 2 \lambda \norm{\beta}_1^2 + \frac{2u^T}{n}\left( \z - \X \beta \right)\right\}
  \nonumber\\
  &= \min_{\z \in \real^n} \left( \frac{1}{n} \norm{\y - \z}^2_2 + \frac{2}{n} u^T\z \right) + \min_{\beta \in \real^p} \left\{ 2 \lambda \norm{\beta}_1^2 - 2 \left( \frac{X^T u}{n} \right)^T \beta \right\}.
  \nonumber
\end{align}
The minimization of $u$ is
\begin{align}
  \min_{\z \in \real^n} \left(\frac{1}{n} \norm{\y - \z}^2_2 + \frac{2}{n} u^T\z \right) = \frac{2}{n} u^T \y - \frac{1}{n} \norm{u}^2_2 = \frac{1}{n} \left( \norm{\y}^2_2 - \norm{\y - u}^2_2\right),
  \nonumber
\end{align}
where the minimum is attained at 
\begin{align}
  \hat{\z} = \y - u.
  \nonumber
\end{align}
The minimization problem of $\beta$ can be written as
\begin{align}
  \min_{\beta \in \real^p} \left\{ 2 \lambda \norm{\beta}_1^2 - 2 \left( \frac{\X^T u}{n} \right)^T \beta \right\} 
  = -2\lambda \max_{\beta \in \real^p} \left\{ \left( \frac{\X^T u}{\lambda n} \right)^T \beta - \norm{\beta}_1^2 \right\}.
  \nonumber
\end{align}
Observe that the maximum is the Fenchel conjugate function of $\norm{\cdot}_1^2$, evaluated at $(\lambda n)^{-1}\X^T u$.
By \citet[Example~3.27, pp.~92-93]{boyd2004convex},
\begin{align}
  -2\lambda \max_{\beta \in \real^p} \left\{ \left( \frac{\X^T u}{\lambda n} \right)^T \beta - \norm{\beta}_1^2 \right\}
  = -\frac{2\lambda}{4} \norm{\frac{\X^T u}{\lambda n}}_\infty^2 = -\frac{1}{2\lambda} \norm{\frac{\X^T u}{n}}_\infty^2.
  \nonumber
\end{align}
So 
\begin{align}
  g\left( u \right) = \frac{1}{n} \left( \norm{\y}^2_2 - \norm{\y - u}^2_2 \right) - \frac{1}{2\lambda} \norm{\frac{\X^Tu}{n}}_\infty^2.
  \nonumber
\end{align}

\section{Proof of Lemma \ref{lem:close_to_oracle_organic}}
A direct upper bound is
\begin{align}
  \check{\sigma}^2_\lambda \leq \frac{1}{n} \norm{\y - \X\beta^\ast}^2_2 + 2\lambda \norm{\beta^\ast}_1^2 = \frac{1}{n} \norm{\varepsilon}^2_2 + 2\lambda \norm{\beta^\ast}_1^2.
  \nonumber
\end{align}
To get a lower bound of $\hat{\sigma}^2$, note that the dual problem in Lemma \ref{lem:dual} and the strong duality imply that
\begin{align}
  \check{\sigma}^2_\lambda &= \min_{\beta \in \real^p} \left( \frac{1}{n} \norm{\y - \X \beta}^2_2 + 2 \lambda \norm{\beta}_1^2 \right)
   = \max_{u \in \real^n} \left( \frac{1}{n}\norm{\y}^2_2 - \frac{1}{n}\norm{\y - u}^2_2 - \frac{1}{2\lambda} \norm{\frac{\X^Tu}{n}}_\infty^2 \right)
  \nonumber\\
  & \geq \frac{1}{n}\norm{\y}^2_2 - \frac{1}{n}\norm{\y - \varepsilon}^2_2 - \frac{1}{2\lambda} \norm{\frac{\X^T \varepsilon}{n}}_\infty^2 
   = \frac{1}{n} \norm{\varepsilon}^2_2 + \frac{2}{n} \varepsilon^T \X \beta^\ast - \frac{1}{2\lambda} \norm{\frac{\X^T \varepsilon}{n}}_\infty^2 
  \nonumber\\
  & \geq \frac{1}{n} \norm{\varepsilon}^2_2 - 2 \norm{\frac{\X^T\varepsilon}{n}}_\infty \norm{\beta^\ast}_1 - \frac{1}{2\lambda} \norm{\frac{\X^T \varepsilon}{n}}_\infty^2
   \geq \frac{1}{n} \norm{\varepsilon}^2_2 - 2 \lambda \sigma^2 \left( \frac{\norm{\beta^\ast}_1}{\sigma} + \frac{1}{4} \right),
  \nonumber
\end{align}
where the last inequality holds for 
\begin{align}
  \lambda \geq \frac{\norm{\X^T \varepsilon}_\infty}{n\sigma}.
  \nonumber
\end{align}

\section{Proof of Theorem \ref{thm:msebound} and Theorem \ref{thm:msebound_olasso}}
We present in this section the proof of Theorem \ref{thm:msebound_olasso}. The proof of Theorem \ref{thm:msebound} follows
the same set of arguments.
First we use the following lemma to characterize the event that $\lambda \geq n^{-1}\sigma^{-1}\snorm{\X^T \varepsilon}_\infty$ is true, so that we can use Lemma \ref{lem:close_to_oracle_organic} to prove a high probability bound.
\begin{lemma}[Corollary 4.3, \citet{giraud2014introduction}] \label{lem:lambda}
  Assume that each column $\X_j$ of the design matrix $\X \in \real^{n \times p}$ satisfies $\snorm{\X_j}^2_2 = n$ for all $j = 1, \dots, p$, and
  $\varepsilon \sim N \left( \0, \sigma^2 \iden_n \right)$. Then for any $L > 0$, 
  \begin{align}
    \Prob\left\{ \frac{\norm{\X^T \varepsilon}_\infty}{n\sigma} >  \left( \frac{2 \log p + 2L}{n} \right)^{1/2} \right\} \leq e^{-L}.
    \nonumber
  \end{align}
\end{lemma}

Lemma \ref{lem:lambda} implies that a good choice of the value of $\lambda$ would be $\{n^{-1}(2\log p + 2L)\}^{1/2}$, which does not depend on
any parameter of the underlying model.
The following corollary shows that with this value of $\lambda$, the organic lasso estimate of $\sigma^2$ is close to the oracle estimator with high probability.

\begin{corollary} \label{cor:hpbound}
  Assume that each column $\X_j$ of the design matrix $\X \in \real^{n \times p}$ satisfies $\snorm{\X_j}^2_2 = n$ for all $j = 1, \dots, p$, and
  $\varepsilon \sim N \left( \0, \sigma^2 \iden_n \right)$. Then for any $L > 0$, the organic lasso with
  \begin{align}
    \lambda = \left( \frac{2 \log p + 2L}{n} \right)^{1/2}
    \nonumber
  \end{align}
  has the following bound
  \begin{align}
    \left( \check{\sigma}^2_\lambda - \frac{1}{n} \norm{\varepsilon}^2_2 \right)^2 \leq 
    8 \max\left\{ \norm{\beta^\ast}_1^2, \sigma^2 \left( \frac{\norm{\beta^\ast}_1}{\sigma} + \frac{1}{4} \right) \right\}^2 \frac{\log p + L}{n}
    \nonumber
  \end{align}
  with probability greater than $1 - e^{-L}$.
\end{corollary}

In general, a high probability bound does not necessarily imply an expectation bound. However, when the probability bound holds with an
exponential tail, it implies an expectation bound with essentially the same rate.
\begin{theorem} \label{thm:expbound}
  Assume that each column $\X_j$ of the design matrix $\X \in \real^{n \times p}$ satisfies $\norm{\X_j}^2_2 = n$ for all $j = 1, \dots, p$, and
  $\varepsilon \sim N \left( \0, \sigma^2 \iden_n \right)$. Then, for any constant $M > 1$,
  the organic lasso estimate with
  \begin{align}
    \lambda = \left( \frac{2M\log p}{n} \right)^{1/2}
    \nonumber
  \end{align}
   satisfies the following bound in expectation:
  \begin{align}
    \E \left\{ \left( \check{\sigma}^2_\lambda - \frac{1}{n} \norm{\varepsilon}^2_2 \right)^2 \right\} \leq 8\left( M + \frac{p^{1 - M}}{\log p} \right) \max\left\{ \norm{\beta^\ast}_1^2, \sigma^2 \left( \frac{\norm{\beta^\ast}_1}{\sigma} + \frac{1}{4} \right) \right\}^2 \frac{\log p}{n}.
    \nonumber
  \end{align}
\end{theorem}
\begin{proof}
  For any $M > 1$, take $L = ( M - 1 ) \log p$ in Corollary \ref{cor:hpbound}.
  Denote $X_n = ( \check{\sigma}^2_\lambda - n^{-1} \snorm{\varepsilon}^2 )^2$, and 
  $r_n = 8 \max( \norm{\beta^\ast}_1^2, \sigma \norm{\beta^\ast}_1 + 4^{-1} \sigma^2 )^2 n^{-1} \log p$.
  Then we have
  \begin{align}
    \Prob \left( X_n > M r_n\right) \leq e^{-\left( M - 1 \right) \log p}.
    \nonumber
  \end{align}
  So
    \begin{align}
    \E \left( \frac{X_n}{r_n} \right) &= \int_{0}^\infty \Prob\left( \frac{X_n}{r_n} > t \right) \d t
    = \int_{0}^{M} \Prob\left( \frac{X_n}{r_n} > t \right) \d t + \int_{M}^\infty \Prob\left( \frac{X_n}{r_n} > t \right) \d t
    \nonumber\\
    &\leq M + \int_{M}^\infty e^{-\left( t - 1 \right) \log p} \d t
    = M + \frac{p^{1 - M}}{\log p},
    \nonumber
  \end{align}
  and the expectation bound follows.
\end{proof}
Now we are ready to present the proof of Theorem \ref{thm:msebound_olasso}.
Since $\sigma^{-2} \norm{\varepsilon}^2_2 \sim \chi^2(n)$, we have 
  \begin{align}
    \E\left( \frac{1}{n} \norm{\varepsilon}^2_2 \right) = \sigma^2, \qquad \Var\left( \frac{1}{n} \norm{\varepsilon}^2_2 \right) = \frac{2\sigma^4}{n},
    \nonumber
  \end{align}
  Therefore, 
  \begin{align}
    &\E \left\{ \left( \check{\sigma}^2_\lambda - \sigma^2 \right)^2 \right\} = 
    \E \left\{ \left( \check{\sigma}^2_\lambda - \frac{1}{n}\norm{\varepsilon}^2_2 + \frac{1}{n}\norm{\varepsilon}^2_2 - \sigma^2 \right)^2 \right\}
    \nonumber\\
    & = \E \left\{ \left( \check{\sigma}^2_\lambda - \frac{1}{n}\norm{\varepsilon}^2_2 \right)^2\right\} + \E \left\{ \left(\frac{1}{n}\norm{\varepsilon}^2_2 - \sigma^2 \right)^2 \right\}
    + 2 \E \left\{ \left( \check{\sigma}^2_\lambda - \frac{1}{n}\norm{\varepsilon}^2_2 \right)\left( \frac{1}{n}\norm{\varepsilon}^2_2 - \sigma^2 \right) \right\}
    \nonumber\\
    & \leq \E \left\{ \left( \check{\sigma}^2_\lambda - \frac{1}{n}\norm{\varepsilon}^2_2 \right)^2\right\} + \Var \left(\frac{1}{n}\norm{\varepsilon}^2_2 \right)
    + 2 \left\{ \Var \left( \check{\sigma}^2_\lambda - \frac{1}{n}\norm{\varepsilon}^2_2 \right) \Var \left( \frac{1}{n}\norm{\varepsilon}^2_2 \right) \right\}^{1/2}
    \nonumber\\
    & \leq \E \left\{ \left( \check{\sigma}^2_\lambda - \frac{1}{n}\norm{\varepsilon}^2_2 \right)^2\right\} + \Var \left(\frac{1}{n}\norm{\varepsilon}^2_2 \right)
    + 2 \left[ \E \left\{ \left( \check{\sigma}^2_\lambda - \frac{1}{n}\norm{\varepsilon}^2_2 \right)^2 \right\} \Var \left( \frac{1}{n}\norm{\varepsilon}^2_2 \right) \right]^{1/2}
    \nonumber\\
    & = \left[ \left[ \E \left\{ \left( \check{\sigma}^2_\lambda - \frac{1}{n}\norm{\varepsilon}^2_2 \right)^2 \right\} \right]^{1/2} + \left\{ \Var \left(\frac{1}{n}\norm{\varepsilon}^2_2 \right) \right\}^{1/2} \right]^2
    \nonumber\\
    & \leq \left[ \left\{ 8 \left( M + \frac{p^{1 - M}}{\log p} \right) \right\}^{1/2} \max\left\{ \norm{\beta^\ast}_1^2, \sigma^2 \left( \frac{\norm{\beta^\ast}_1}{\sigma} + \frac{1}{4} \right) \right\} \left( \frac{\log p}{n} \right)^{1/2} + \sigma^2 \left( \frac{2}{n} \right)^{1/2} \right]^2,
    \nonumber
  \end{align}
  where the last inequality holds from Theorem \ref{thm:expbound}.

  \section{Proof of Remark \ref{rem:polynomialmoment}}
  For the independent zero-mean noise $\varepsilon_i$ with variance $\sigma^2$ and bounded $m$-th order moment ($m = 3, 4, \dots$)
  \begin{align}
    \E |\varepsilon_i|^m \leq \frac{m!}{2} K^{m - 2}
    \nonumber
  \end{align}
  for some constant $K > 0$, a Bernstein's type inequality \citep[Lemma 14.13]{buhlmann2011statistics} implies that
  \begin{align}
    \Prob\left[ \max_{1 \leq j \leq p} \frac{1}{n\sigma} \norm{\X_j^T \varepsilon}_\infty \geq \frac{2K \log p}{n} + 2 \left\{ \frac{\log(2p)}{n} \right\}^{1/2} \right] \leq \frac{1}{p}.
    \nonumber
  \end{align}
  Then the proof of Corollary \ref{cor:hpbound} goes through.

  \section{Proof of Proposition \ref{prop:slow_rate_naive} and Proposition \ref{prop:slow_rate_sqrt}}
  The following lemma gives a general result on the estimation error of $\hat{\sigma}^2$ of the form \eqref{est:naive} in the paper based on $\hat{\beta}$:
  \begin{lemma} \label{lemma:basic}
    \begin{align}
      \left| \hat{\sigma}^2 - \frac{1}{n} \snorm{\varepsilon}_2^2 \right| \leq \frac{1}{n} \norm{\X \hat{\beta} - \X \bt}_2^2 + \frac{2}{n} \norm{\X^T \varepsilon}_\infty \left( \snorm{\bt}_1 + \snorm{\hat{\beta}}_1 \right)
      \nonumber
    \end{align}
  \end{lemma}
  \begin{proof}
    First by definition
    \begin{align}
      \hat{\sigma}^2 = \frac{1}{n} \norm{\y - \X \hat{\beta}}_2^2 = \frac{1}{n}\norm{\varepsilon + \X \bt - \X \hat{\beta}}_2^2 = \frac{1}{n} \snorm{\varepsilon}_2^2 + \frac{1}{n} \norm{\X \hat{\beta} - \X \bt}_2^2 + \frac{2}{n} \varepsilon^T \X \left( \hat{\beta} - \bt \right).
      \nonumber
    \end{align}
    Note that
    \begin{align}
      \left| \varepsilon^T \X \left( \hat{\beta} - \bt \right)\right| \leq \snorm{\X^T \varepsilon}_\infty \snorm{\hat{\beta} - \bt}_1,
      \nonumber
    \end{align}
    and the result follows.
  \end{proof}
  \subsection{Slow rate bound for the naive estimator of $\sigma^2$}
  We now give the proof of Proposition \ref{prop:slow_rate_naive}. 
    From the basic inequality
    \begin{align}
      \frac{1}{n} \norm{\y - \X \be}_2^2 + 2 \lambda \snorm{\be}_1 \leq \frac{1}{n} \norm{\y - \X \bt}_2^2 + 2 \lambda \snorm{\bt}_1,
      \nonumber
    \end{align}
    which implies that 
    \begin{align}
      \frac{1}{n} \norm{\X \be - \X \bt}_2^2 + 2 \lambda \snorm{\be}_1  &\leq \frac{2}{n} \left| \varepsilon^T \X \left( \be - \bt \right) \right| + 2 \lambda \snorm{\bt}_1  \nonumber\\
      &\leq \frac{2}{n} \norm{\X^T \varepsilon}_\infty \norm{\be - \bt}_1 + 2\lambda \snorm{\bt}_1.
      \nonumber
    \end{align}
    We thank Irina Gaynanova \citep{irina18commu} for showing us the technique of taking $\lambda$ to be twice its usual size. For $\lambda \geq 2 n^{-1} \snorm{\X^T\varepsilon}_\infty $, we have that
    \begin{align}
      \frac{1}{n} \norm{\X \be - \X \bt}_2^2 + 2\lambda \snorm{\be}_1 \leq 
      \lambda \snorm{\be - \bt}_1 + 2 \lambda \snorm{\bt}_1 \leq \lambda \snorm{\be}_1 + 3 \lambda \snorm{\bt}_1,
      \nonumber
    \end{align}
    so $n^{-1} \snorm{\X \be - \X \bt}_2^2 + \lambda \snorm{\be}_1 \leq 3 \lambda \snorm{\bt}_1$.
    So by Lemma \ref{lemma:basic} we have
    \begin{align}
      \left| \se - \frac{1}{n} \snorm{\varepsilon}_2^2 \right|    
      &\leq \frac{1}{n} \norm{\X \be - \X \bt}_2^2 + \frac{2}{n} \norm{\X^T \varepsilon}_\infty \left( \snorm{\bt}_1 + \snorm{\be}_1 \right) \nonumber\\
      &\leq \frac{1}{n} \norm{\X \be - \X \bt}_2^2 + \lambda \snorm{\bt}_1 + \lambda \snorm{\be}_1
      \leq 4 \lambda \snorm{\bt}_1.
      \nonumber
    \end{align}
    Finally, taking $\lambda = 2 \sigma \{n^{-1}(2 \log p + 2 L)\}^{1/2}$ with $L = \log p$, the result follows from Lemma \ref{lem:lambda}.

  \subsection{Slow rate bound for the square-root/scaled lasso estimator of $\sigma^2$}

  As shown in \citet{2016arXiv160800624L} (proof of Lemma A.3), we note that with probability $1$, $\snorm{y - X \bsqrt}_2 > 0$ for $\lambda > 0$. 
  So the first order optimality condition of the square-root/scaled lasso is
  \begin{align}
    \frac{1}{n^{1/2}} \frac{-\X^T\left( \y - \X \bsqrt \right)}{\norm{\y - \X \bsqrt}_2} + \lambda \hat{g} = 0
      \nonumber
    \end{align}
    for some $\hat{g} \in \partial\snorm{\bsqrt}_1$. Taking an inner product with $\bsqrt - \bt$ on both sides, we have
    \begin{align}
      -\frac{1}{n^{1/2}}\frac{\left( \bsqrt - \bt \right)^T \X^T \left( \y - \X \bsqrt \right)}{\norm{\y - \X \bsqrt}_2} + \lambda \hat{g}^T \left( \bsqrt - \bt \right) = 0,
      \nonumber
    \end{align}
    which implies that
    \begin{align}
      \frac{\norm{\X\left( \bt - \bsqrt \right)}_2^2}{n^{1/2} \norm{\y - \X \bsqrt}_2} - \frac{\left( \bsqrt - \bt \right)^T \X^T \varepsilon}{n^{1/2} \norm{\y - \X \bsqrt}_2} \leq \lambda \hat{g}^T \left( \bt - \bsqrt \right) \leq \lambda \snorm{\bt}_1 - \lambda \snorm{\bsqrt}_1,
      \nonumber
    \end{align}
    and thus
    \begin{align}
      &\frac{1}{n} \norm{\X \left( \bt - \bsqrt \right)}_2^2 \leq \frac{1}{n} \left|\varepsilon^T \X \left( \bsqrt - \bt \right) \right|
      + \frac{\lambda}{n^{1/2}} \norm{\y - \X \bsqrt}_2 \left( \snorm{\bt}_1 - \snorm{\bsqrt}_1 \right) \nonumber\\
      \leq& \frac{1}{n} \norm{\X^T \varepsilon}_\infty \norm{\bsqrt - \bt}_1 + \frac{\lambda}{n^{1/2}} \norm{\y - \X \bsqrt}_2 \left( \snorm{\bt}_1 - \snorm{\bsqrt}_1 \right) \nonumber \\
      \leq& \frac{1}{n} \norm{\X^T \varepsilon}_\infty \left( \snorm{\bsqrt}_1 + \snorm{\bt}_1 \right) + \frac{\lambda}{n^{1/2}} \norm{\y - \X \bsqrt}_2 \left( \snorm{\bt}_1 - \snorm{\bsqrt}_1 \right) \nonumber \\
      \leq& \left( \frac{1}{n} \norm{\X^T \varepsilon}_\infty + \frac{\lambda}{n^{1/2}} \norm{\y - \X \bsqrt}_2 \right) \snorm{\bt}_1 + \left( \frac{1}{n} \norm{\X^T \varepsilon}_\infty - \frac{\lambda}{n^{1/2}} \norm{y - X \bsqrt}_2 \right) \snorm{\bsqrt}_1.
      \nonumber
    \end{align}
    Taking $\lambda = 3 n^{-1/2} \snorm{y - X \bsqrt}_2^{-1} \snorm{X^T \varepsilon}_\infty$, which is 3 times what is suggested in \citet{2016arXiv160800624L}, we have
    \begin{align}
      \frac{1}{n} \norm{\X \left( \bt - \bsqrt \right)}_2^2 \leq \frac{4 \norm{\X^T \varepsilon}_\infty}{n} \snorm{\bt}_1 - \frac{2 \snorm{X^T\varepsilon}_\infty}{n} \snorm{\bsqrt}_1.
      \nonumber
    \end{align}
    By Lemma \ref{lemma:basic}
    \begin{align}
      \left| \ssqrt - \frac{1}{n} \snorm{\varepsilon}_2^2 \right| &\leq \frac{1}{n} \norm{\X \bsqrt - \X \bt}_2^2 + \frac{2}{n} \norm{\X^T \varepsilon}_\infty \left( \snorm{\bt}_1 + \snorm{\bsqrt}_1 \right) \nonumber\\
      &\leq  \frac{6}{n} \norm{\X^T \varepsilon}_\infty \snorm{\bt}_1.
      \nonumber
    \end{align}
    The result then follows from Lemma \ref{lem:lambda} by taking $L = \log p$.

\section{Proof of Proposition \ref{thm:scaleequivariant}: scale-equivariance of the organic lasso}
\begin{proof}
  Suppose $\check{\beta}_\lambda \left( \y\right)$ is a solution to the organic lasso, where we write out explicitly the dependence of the solution on the response $\y$. Then using notation from previous section,
  \begin{align}
    L\left( t \check{\beta}_\lambda \left( \y \right) | t\y, \lambda\right) &= \frac{1}{n} \norm{t\y - t\X \check{\beta}_\lambda \left( \y\right)}^2_2 + 2\lambda\norm{t\check{\beta}_\lambda \left( \y\right)}_1^2 
    \nonumber\\
    & = t^2 L\left( \check{\beta}_\lambda \left( \y\right) | \y, \lambda \right).
    \nonumber
  \end{align}
  This implies that $t\check{\beta}_\lambda \left( \y\right)$ is a solution to the problem with response $t\y$, i.e., $\check{\beta}_\lambda \left( t\y\right) = t\check{\beta}\left( \y\right)$.
  Consequently, 
  \begin{align}
    \check{\sigma}^2_\lambda \left( t\y \right) &= \min_\beta L\left( \beta_\lambda | t \y, \lambda \right) \nonumber\\
    &= L\left( t\check{\beta}_\lambda \left( \y, \lambda \right) | t\y, \lambda\right) 
    = t^2 L \left( \check{\beta}_\lambda \left( \y, \lambda \right) | \y, \lambda \right) 
    = t^2\check{\sigma}^2_\lambda \left( \y, \lambda \right),
    \nonumber
  \end{align}
  which establishes the theorem.
\end{proof}

\section{Proof of Theorem \ref{thm:olasso_predict}} \label{proof:olasso_predict}
\begin{proof}
  We start from the basic inequality
  \begin{align}
    \frac{1}{n} \norm{\y - \X \check{\beta}_\lambda}_2^2 + 2\lambda \norm{\check{\beta}_\lambda}_1^2 \leq \frac{1}{n} \norm{\y - \X \beta^\ast}_2^2 + 2 \lambda \norm{\beta^\ast}_1^2,
    \nonumber
  \end{align}
  which leads to 
  \begin{align}
    \frac{1}{n} \norm{\X \check{\beta}_\lambda - \X \beta^\ast}_2^2 &\leq 2 \left( \frac{\X^T \varepsilon}{n} \right)^T \left( \check{\beta}_\lambda - \beta^\ast \right) + 2 \lambda \left( \norm{\beta^\ast}_1^2 - \norm{\check{\beta}_\lambda}_1^2 \right)
    \nonumber\\
    &\leq 2 \norm{\frac{\X^T \varepsilon}{n}}_\infty \norm{\check{\beta}_\lambda - \beta^\ast}_1 + 2 \lambda \left( \norm{\beta^\ast}_1^2 - \norm{\check{\beta}_\lambda}_1^2 \right).
    \nonumber
  \end{align}
  If 
  \[
    \norm{\frac{\X^T \varepsilon}{n}}_\infty \leq \sigma \lambda,
  \]
  then
  \begin{align}
    \frac{1}{n} \norm{\X \check{\beta}_\lambda - \X \beta^\ast}_2^2 &\leq 2 \sigma \lambda \norm{\check{\beta}_\lambda - \beta^\ast}_1 + 2 \lambda \left( \norm{\beta^\ast}_1^2 - \norm{\check{\beta}_\lambda}_1^2 \right)
    \nonumber\\
    &\leq \sigma^2 \lambda + \lambda \norm{\check{\beta}_\lambda - \beta^\ast}_1^2 + 2 \lambda \left( \norm{\beta^\ast}_1^2 - \norm{\check{\beta}_\lambda}_1^2 \right)
    \nonumber\\
    &\leq \sigma^2 \lambda + \lambda \left( \norm{\check{\beta}_\lambda}_1 + \norm{\beta^\ast}_1 \right)^2 + 2 \lambda \left( \norm{\beta^\ast}_1^2 - \norm{\check{\beta}_\lambda}_1^2 \right)
    \nonumber\\
    &\leq \sigma^2 \lambda + 2 \lambda \left( \norm{\check{\beta}_\lambda}_1^2 + \norm{\beta^\ast}_1^2 \right) + 2 \lambda \left( \norm{\beta^\ast}_1^2 - \norm{\check{\beta}_\lambda}_1^2 \right)
    \nonumber\\
    &= \sigma^2 \lambda + 4 \lambda \norm{\beta^\ast}_1^2.
    \nonumber
  \end{align}
  The result then holds from Lemma \ref{lem:lambda}.
\end{proof}

\section{Mapping between the paths of the natural and organic lasso} \label{proof:equivalence}

In this section, we draw a connection between the natural lasso and the organic lasso estimates of $\beta^\ast$.
\begin{theorem} \label{thm:equivalence}
  Letting $\hat{\beta}_s$ and $\check{\beta}_t$ denote the lasso and organic
  lasso estimates of $\beta^\ast$ with tuning parameters $s$ and $t$,
  \begin{align}
    \hat{\beta}_\lambda = \check{\beta}_{( 2\snorm{\hat{\beta}_\lambda}_1 )^{-1}\lambda}, \qquad 
    \check{\beta}_\nu = \hat{\beta}_{2 \nu \snorm{\check{\beta}_\nu}_1}.
  \end{align}
\end{theorem}
This result implies that one can start with a lasso solution $\hat{\beta}_\lambda$ with tuning parameter $\lambda$,
and then report a solution to the organic lasso with tuning parameter $( 2 \snorm{\hat{\beta}_\lambda}_1)^{-1} \lambda$.
Likewise, an organic lasso solution $\check{\beta}_\nu$ is equivalent to a standard lasso solution with tuning parameter
$2 \nu \snorm{\check{\beta}_\nu}_1$. This equivalence is also observed
in \citet{lorbert2010exploiting} that considers a more general
penalty.

Although the methods' paths are the same, this does not imply that the
cross-validated methods will be the same.
In $K$-fold cross-validation, the natural lasso estimator is evaluated on $K$
differing datasets for a fixed value of $\lambda$.
A fixed tuning parameter $\lambda$ for the natural lasso over multiple datasets corresponds to running the organic lasso with a different $\lambda$ on each fold.
Thus, the two methods in fact have different cross-validation performance.

\begin{proof}
  Let $\hat{\beta}_\lambda$ be a solution to \eqref{est:lasso} with tuning parameter $\lambda$, and $\tilde{\beta}_\nu$ be a solution
  to \eqref{est:beta_olasso} with tuning parameter $\nu$, then they satisfy optimality conditions
  \begin{align}
    -\frac{1}{n} \X^T \left( \y - \X \hat{\beta}_\lambda \right) + \lambda \hat{g} = \0 \qquad \text{where} \qquad \hat{g} \in \partial \left( \norm{\hat{\beta}_\lambda}_1 \right),
    \label{eq:opt_lasso}\\
    -\frac{1}{n} \X^T \left( \y - \X \tilde{\beta}_\nu \right) + 2 \nu \norm{\tilde{\beta}_\nu}_1 \tilde{g} = \0 \qquad \text{where} \qquad \tilde{g} \in \partial \left( \norm{\tilde{\beta}_\nu}_1 \right).
    \label{eq:opt_olasso}
  \end{align}

  If $\hat{\beta}_\lambda = \tilde{\beta}_\nu$, then simply comparing \eqref{eq:opt_lasso} and \eqref{eq:opt_olasso} we have that
  $\lambda = 2 \nu \snorm{\tilde{\beta}_\nu}_1$, and $\nu = ( 2 \snorm{\hat{\beta}_\lambda}_1)^{-1} \lambda$.

  Now for $\hat{\beta}_\lambda$ that satisfies \eqref{eq:opt_lasso}, by plugging $\lambda = 2\nu \snorm{\hat{\beta}_\lambda}_1$,
  we have that $\hat{\beta}_\lambda$ satisfies \eqref{eq:opt_olasso},
  i.e., $\tilde{\beta}_\nu = \hat{\beta}_\lambda$ where $\lambda = 2\nu \snorm{\hat{\beta}_\lambda}_1$.
  Following the same argument, for $\tilde{\beta}_\nu$ that satisfies \eqref{eq:opt_olasso},
  we take $\nu = (2 \snorm{\tilde{\beta}_\nu}_1)^{-1} \lambda$, and find that $\tilde{\beta}_\nu$ satisfies \eqref{eq:opt_lasso}.
  This implies that $\hat{\beta}_\lambda = \tilde{\beta}_\nu$, where $\nu = (2 \snorm{\tilde{\beta}_\nu}_1)^{-1} \lambda$.
\end{proof}

\section{Fast rate in prediction error of the squared lasso} \label{proof:fastrate}
Recall the squared lasso estimate of $\beta^\ast$:
\begin{align}
  \ob \in \argmin_{\beta \in \real^p} \frac{1}{n} \norm{\y - \X \beta}_2^2 + 2 \lambda \norm{\beta}_1^2.
  \label{eq:olasso}
\end{align}
It is well known that the fast rate is built on the compatibility condition of the lasso problem. Let $\S = \mathrm{supp}(\tb)$, i.e., the support of the true regression coefficient $\beta^\ast$, the compatibility condition of the squared lasso problem requires that for all $\mu \in \real^p$ such that $\snorm{\mu_{\S^c}}_1 - \sigma \leq 3 \snorm{\mu_\S}_1$, 
\begin{align}
  \snorm{\mu_\S}_1 + \frac{1}{4} \sigma \leq |\S|^{1/2} \frac{\norm{\X \mu}_2}{n^{1/2}\phi_0}.
  \label{eq:compatability_olasso}
\end{align}

The following theorem establishes that the fast rate prediction error and an estimation error rate of $\check{\beta}$ in \eqref{eq:olasso} can be attained with a value of $\lambda$ that does not depend on any unknown parameters.
\begin{theorem}
  Suppose that each column $\X_j$ of the matrix $\X \in \real^{n \times p}$ has been scaled so that $\snorm{\X_j}^2_2 = n$ for all $j = 1, \ldots, p$, and $\varepsilon \sim N \left( \0, \sigma^2 I_n \right)$. If compatibility condition \eqref{eq:compatability_olasso} holds, then for any $L > 0$, the solution $\ob$ in \eqref{eq:olasso} with tuning parameter 
  \begin{align}
    \lambda = \left( \frac{2 \log p + 2 L}{n} \right)^{1/2}
    \label{eq:lam_olasso}
  \end{align}
  attains the following estimation error rate and fast rate bound in prediction with probability greater than $1 - e^{-L}$:
  \begin{align}
    &\frac{1}{2n} \norm{\X \ob - \X \tb}_2^2
    \leq \frac{64 \max\left(\norm{\tb}_1, \sigma \right)^2 |\S| \left( \log p + L \right)}{\phi_0^2 n};
    \nonumber \\
    &\norm{\tb - \ob}_1
    \leq\frac{16 \max\left(\norm{\tb}_1, \sigma  \right) |\S|}{\phi_0^2} \left( \frac{2 \log p + 2 L}{n} \right)^{1/2}.
    \nonumber
  \end{align}
\end{theorem}

\begin{proof}
  First by the optimality of $\ob$, we have
  \begin{align}
    \frac{1}{n} \norm{\y - \X \ob}_2^2 + 2\lambda \norm{\ob}_1^2 \leq \frac{1}{n} \norm{\y - \X \tb}_2^2 + 2 \lambda \norm{\tb}_1^2,
    \nonumber
  \end{align}
  which implies that
  \begin{align}
    \frac{1}{n} \norm{\X \ob - \X \tb}_2^2 \leq \frac{2}{n}\left( \ob - \tb \right)^T \X^T \varepsilon + 2 \lambda \norm{\tb}_1^2 - 2 \lambda \norm{\ob}_1^2.
    \label{eq:basic}
  \end{align}
  The following proof is considered in two cases:

  (1). When $\norm{\tb}_1 \geq \sigma$:
  Note that $\norm{\cdot}_1^2$ is convex and by chain rule, for any $g \in \partial(\norm{\beta^\ast}_1)$,
  \begin{align}
    \norm{\ob}_1^2 - \norm{\tb}_1^2 \geq 2 \norm{\beta^\ast}_1 g^T \left( \ob - \tb \right).
    \nonumber
  \end{align}
  For $j \in \S$, we have that $g_j = \sign(\tb_j)$. For any $j \in \S^C$, we let
  \begin{align}
    g_j = \sign\left( \ob_j - \tb_j \right) = \sign\left( \ob_j \right).
    \nonumber
  \end{align}
  Then $g$ is still a valid sub-differential of $\norm{\tb}_1$. Moreover, conditional on the event
  \begin{align}
    \T = \left\{ \frac{1}{n} \snorm{\X^T \varepsilon}_\infty \leq \lambda \sigma \right\},
    \nonumber
  \end{align}
  from \eqref{eq:basic} we have
  \begin{align}
    \frac{1}{n} \norm{\X \ob - \X \tb}_2^2 &\leq \frac{2}{n}\left( \ob - \tb \right)^T \X^T \varepsilon + 4 \lambda \norm{\tb}_1 g^T \left( \tb - \ob \right) \nonumber\\
    &= \frac{2}{n}\left( \ob - \tb \right)^T \X^T \varepsilon + 4 \lambda \norm{\tb}_1 g^T_{\S} \left( \tb_\S - \ob_\S \right) + 4 \lambda \norm{\tb}_1 g^T_{\S^C} \left( \tb_{\S^C} - \ob_{\S^C} \right)
    \nonumber \\
    &= \frac{2}{n}\left( \ob - \tb \right)^T \X^T \varepsilon + 4 \lambda \norm{\tb}_1 g^T_{\S} \left( \tb_\S - \ob_\S \right) - 4 \lambda \norm{\tb}_1 \norm{\tb_{\S^C} - \ob_{\S^C}}_1
    \nonumber\\
    &\leq \frac{2}{n}\left( \ob - \tb \right)^T \X^T \varepsilon + 4 \lambda \norm{\tb}_1 \norm{\tb_\S - \ob_\S}_1 - 4 \lambda \norm{\tb}_1 \norm{\tb_{\S^C} - \ob_{\S^C}}_1.
    \nonumber
  \end{align}
  Since $\sigma \leq \snorm{\tb}_1$ and $\T$ holds, we have that $n^{-1} \norm{\X^T \varepsilon}_\infty \leq \lambda \sigma \leq \lambda \norm{\tb}_1$, and thus
  \begin{align}
    \frac{1}{n} \norm{\X \ob - \X \tb}_2^2 &\leq 2 \lambda \norm{\tb}_1 \norm{\tb - \ob}_1 + 4 \lambda \norm{\tb}_1 \norm{\tb_S - \ob_S}_1 - 4 \lambda \norm{\tb}_1 \norm{\tb_{\S^C} - \ob_{\S^C}}_1 \nonumber\\
    & = 2 \lambda \norm{\tb}_1 \left( 3\norm{\tb_\S - \ob_\S}_1 - \norm{\tb_{\S^C} - \ob_{\S^C}}_1 \right).
    \nonumber
  \end{align}
  This first implies that $3\norm{\tb_\S - \ob_\S}_1 \geq \norm{\tb_{\S^C} - \ob_{\S^C}}_1$, and that
  \begin{align}
    \frac{1}{n} \norm{\X \ob - \X \tb}_2^2  + 2 \lambda \norm{\tb}_1 \norm{\tb_{\S^C} - \ob_{\S^C}}_1   \leq  6 \lambda \norm{\tb}_1  \norm{\tb_\S - \ob_\S}_1.
    \nonumber
  \end{align}
  Then by compatibility condition,
  \begin{align}
    &\frac{1}{n} \norm{\X \ob - \X \tb}_2^2 + 2 \lambda \norm{\tb}_1 \norm{\tb - \ob}_1 \nonumber\\
    = &\frac{1}{n} \norm{\X \ob - \X \tb}_2^2 + 2 \lambda \norm{\tb}_1 \norm{\tb_\S - \ob_\S}_1 + 2 \lambda \norm{\tb}_1 \norm{\tb_{\S^C} - \ob_{\S^C}}_1 \nonumber\\
    \leq & 8 \lambda \norm{\tb}_1 \norm{\tb_{\S} - \ob_{\S}}_1 \leq \frac{8 \lambda \norm{\tb}_1 |\S|^{1/2} \norm{\X \tb - \X \ob}_2}{n^{1/2} \phi_0} \nonumber\\
    \leq & \frac{1}{2n} \norm{\X \ob - \X \tb}_2^2 + \frac{32 \norm{\tb}_1^2 \lambda^2 |\S|}{\phi_0^2}.
    \label{eq:bound_smallsig}
  \end{align}

  (2). When $\norm{\tb}_1 < \sigma$: We define $\tg \in \real^p$ as
  \begin{align}
    \tg_j = \begin{cases}
      \tb_j + \frac{\sigma - \norm{\tb}_1}{|\S|} \quad &\text{  if  } \tb_j > 0 \\
      \tb_j - \frac{\sigma - \norm{\tb}_1}{|\S|} \quad &\text{  if  } \tb_j < 0 \\
       0 \quad &\text{  if  } \tb_j = 0.
       \nonumber
    \end{cases}
  \end{align}
  It is easy to check that $\norm{\tg}_1 = \sigma$. Also \eqref{eq:basic} implies that
  \begin{align}
    \frac{1}{n} \norm{\X \ob - \X \tb}_2^2 \leq \frac{2}{n}\left( \ob - \tb \right)^T \X^T \varepsilon + 2 \lambda \left( \norm{\tb}_1^2 - \norm{\tg}_1^2 + \norm{\tg}_1^2 - \norm{\ob}_1^2 \right).
    \nonumber
  \end{align}
  Then we have that
  \begin{align}
    \norm{\ob}_1^2 - \norm{\tg}_1^2 \geq 2 \norm{\tg}_1 g^T \left( \ob - \tg \right)
    \nonumber
  \end{align}
  holds for all $g \in \partial (\snorm{\tg}_1)$, and it further implies that
  \begin{align}
    \norm{\tg}_1^2 - \norm{\ob}_1^2 \leq 2 \norm{\tg}_1 g^T \left( \tg - \ob \right) = 2\sigma g^T \left( \tb - \ob \right) + 2\sigma g^T \left( \tg - \tb \right).
    \nonumber
  \end{align}
  Note that any $g \in \partial (\snorm{\tg}_1)$ is also a valid sub-differential of $\snorm{\tb}_1$, and 
  \begin{align}
    g^T\left( \tg - \tb \right) = \sigma -  \norm{\tb}_1.
    \nonumber
  \end{align}
  Thus we have
  \begin{align}
    \frac{1}{n} \norm{\X \ob - \X \tb}_2^2 &\leq \frac{2}{n}\left( \ob - \tb \right)^T \X^T \varepsilon + 2 \lambda \left( \norm{\tb}_1^2 - \sigma^2 + 2 \sigma g^T \left( \tb - \ob \right) + 2 \sigma^2 - 2 \sigma\norm{\tb}_1 \right) \nonumber\\
    & = \frac{2}{n}\left( \ob - \tb \right)^T \X^T \varepsilon + 4 \lambda \sigma g^T \left( \tb - \ob \right) + 2 \lambda \left( \sigma - \norm{\tb}_1 \right)^2 \nonumber\\
    & \leq \frac{2}{n}\left( \ob - \tb \right)^T \X^T \varepsilon + 4 \lambda \sigma g^T \left( \tb - \ob \right) + 2 \lambda \sigma^2  \nonumber
  \end{align}
  Since $\tg$ and $\tb$ have the same support, we can again choose $g_j = \sign(\ob_j)$ for $j \in S^c$. Conditional on the event $\T$, it follows that
  \begin{align}
    \frac{1}{n} \norm{\X \ob - \X \tb}_2^2 &\leq 2 \lambda \sigma \norm{\ob - \tb}_1 + 4 \lambda\sigma \norm{\ob_\S - \tb_\S}_1 - 4 \lambda \sigma \norm{\ob_{\S^c} - \tb_{\S^c}}_1 + 2 \lambda \sigma^2 \nonumber\\
    &= 6 \lambda\sigma \norm{\ob_\S - \tb_\S}_1 - 2 \lambda \sigma \norm{\ob_{\S^c} - \tb_{\S^c}}_1 + 2 \lambda \sigma^2.
    \nonumber
  \end{align}
  This implies that $3 \snorm{\ob_\S - \tb_\S} + \sigma \geq \snorm{\ob_{\S^c} - \tb_{\S^c}}$. And then by the compatibility condition \eqref{eq:compatability_olasso},
  \begin{align}
    \frac{1}{n} \norm{\X \ob - \X \tb}_2^2 + 2 \lambda \sigma \norm{\ob - \tb}_1 &= \frac{1}{n} \norm{\X \ob - \X \tb}_2^2 + 2 \lambda \sigma \norm{\ob_\S - \tb_\S}_1 + 2\lambda \sigma \norm{\ob_{\S^c} - \tb_{\S^c}}_1 \nonumber\\
    &\leq 8 \lambda \sigma \norm{\ob_{\S} - \tb_{\S}}_1 + 2 \lambda \sigma^2 \nonumber\\
    &\leq 8 \lambda \sigma |\S|^{1/2} \frac{\norm{\X\left( \ob - \tb \right)}_2}{n^{1/2} \phi_0} \nonumber\\
    &\leq \frac{1}{2n} \norm{\X \ob - \X \tb}_2^2 + \frac{32 \lambda^2 \sigma^2 |\S|}{\phi_0^2}.
    \label{eq:bound_bigsig}
  \end{align}
  By the proof of Corollary 4.3 in \citet{giraud2014introduction}, we have
  \begin{align}
    \Prob \left\{ \frac{1}{n} \norm{\X^T \varepsilon}_\infty > \sigma \left( \frac{2 \log p + 2 L }{n} \right)^{1/2} \right\} \leq e^{-L}.
    \nonumber
  \end{align}
  Thus taking $\lambda$ in \eqref{eq:lam_olasso}, we have that
  \begin{align}
    \Prob(\T^c) = \Prob\left( \frac{1}{n} \norm{\X^T \varepsilon}_\infty > \lambda \sigma \right) \leq  e^{-L}.
    \nonumber
  \end{align}
  And the results follow from \eqref{eq:bound_smallsig} and \eqref{eq:bound_bigsig}.
\end{proof}

\section{Additional results in numerical studies} \label{app:numerical}
We include in this section some additional results in the numerical studies in Section \ref{sec:simulation} and Section \ref{sec:realdata}.
In particular, Fig~\ref{fig:cv_appendix} and Fig~\ref{fig:fix_appendix} present the complementary results (in different simulation regimes) to Fig~\ref{fig:cv} and Fig~\ref{fig:fix} in the paper respectively, and Table~\ref{tab:pvalue_t} shows the p-values of the paired t-tests and the Wilcoxon signed-rank tests of the difference of various methods outputs in Fig~\ref{fig:cv} in the paper. Finally, Table~\ref{tab:MSD_ratio} presents the mean and standard errors of $\E(\hat{\sigma} / \sigma)$ of various estimators in the real data example.
\begin{figure}[!htp]
  \centering
  \includegraphics[width = .8\textwidth]{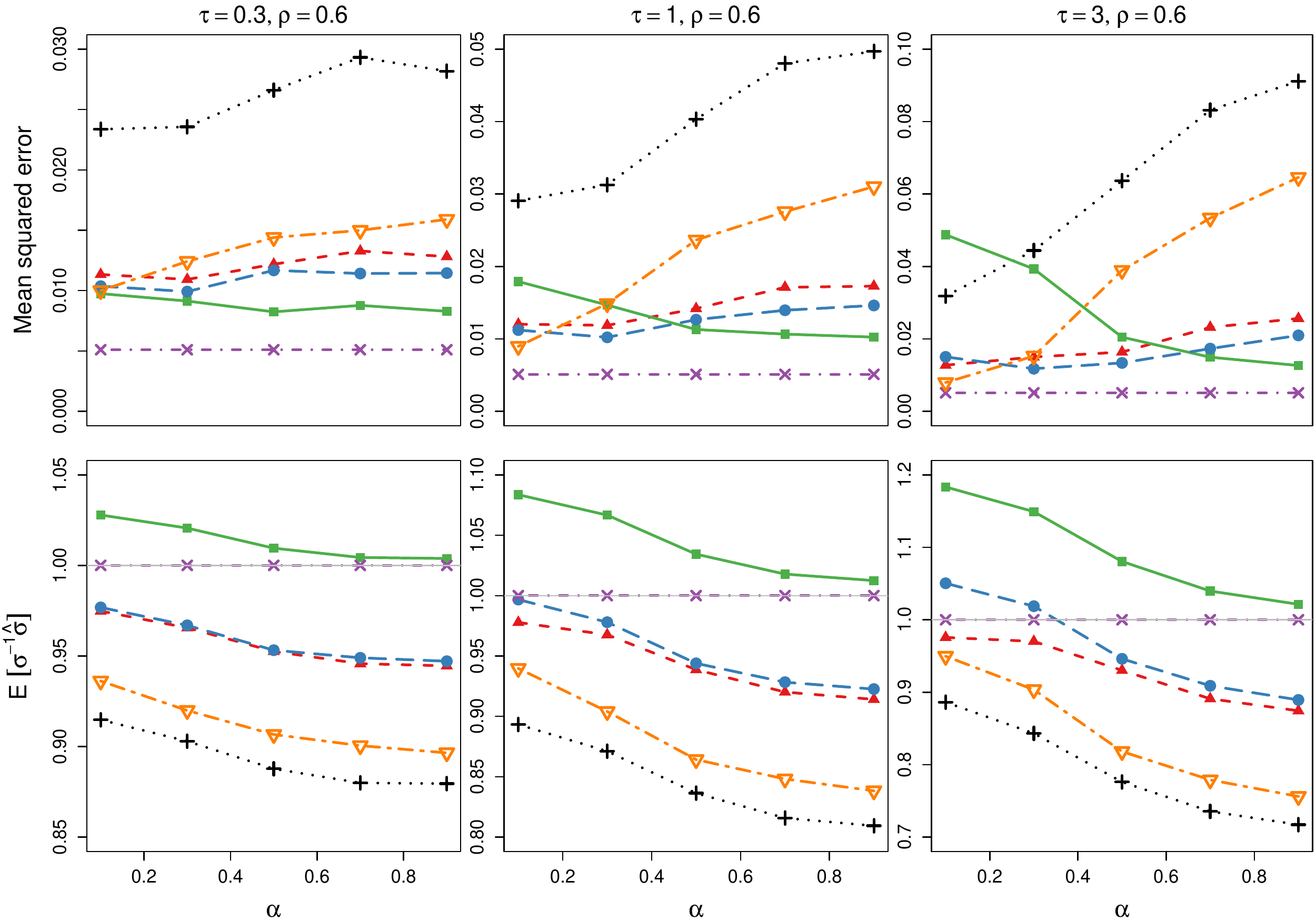}
    \caption{Simulation results of various methods with regularization
      parameter selected using cross-validation. From left to right,
      column show the average (over 1000 repetitions) of the mean
      squared error (top panel) and $E(\sigma^{-1} \hat{\sigma})$ (bottom
      panel) of various methods in three simulation settings. In each
      setting, we fix model sparsity ($\alpha$) and correlations among
      features ($\rho$), and let signal-to-noise ratio(as expressed in
      $\tau$) change. 
      Line styles and their corresponding methods:
         \protect\includegraphics[width = 0.12in]{black.png} for naive, 
   \protect\includegraphics[width = 0.12in]{red.png} for $\rt$, 
   \protect\includegraphics[width = 0.12in]{orange.png} for the square-root/scaled lasso, 
   \protect\includegraphics[width = 0.12in]{green.png} for the natural lasso, 
   \protect\includegraphics[width = 0.12in]{blue.png} for the organic lasso, 
   \protect\includegraphics[width = 0.12in]{purple.png} for the oracle.}
  \label{fig:cv_appendix}
\end{figure}

\begin{figure}[!htp]
  \centering
  \includegraphics[width = .8\textwidth]{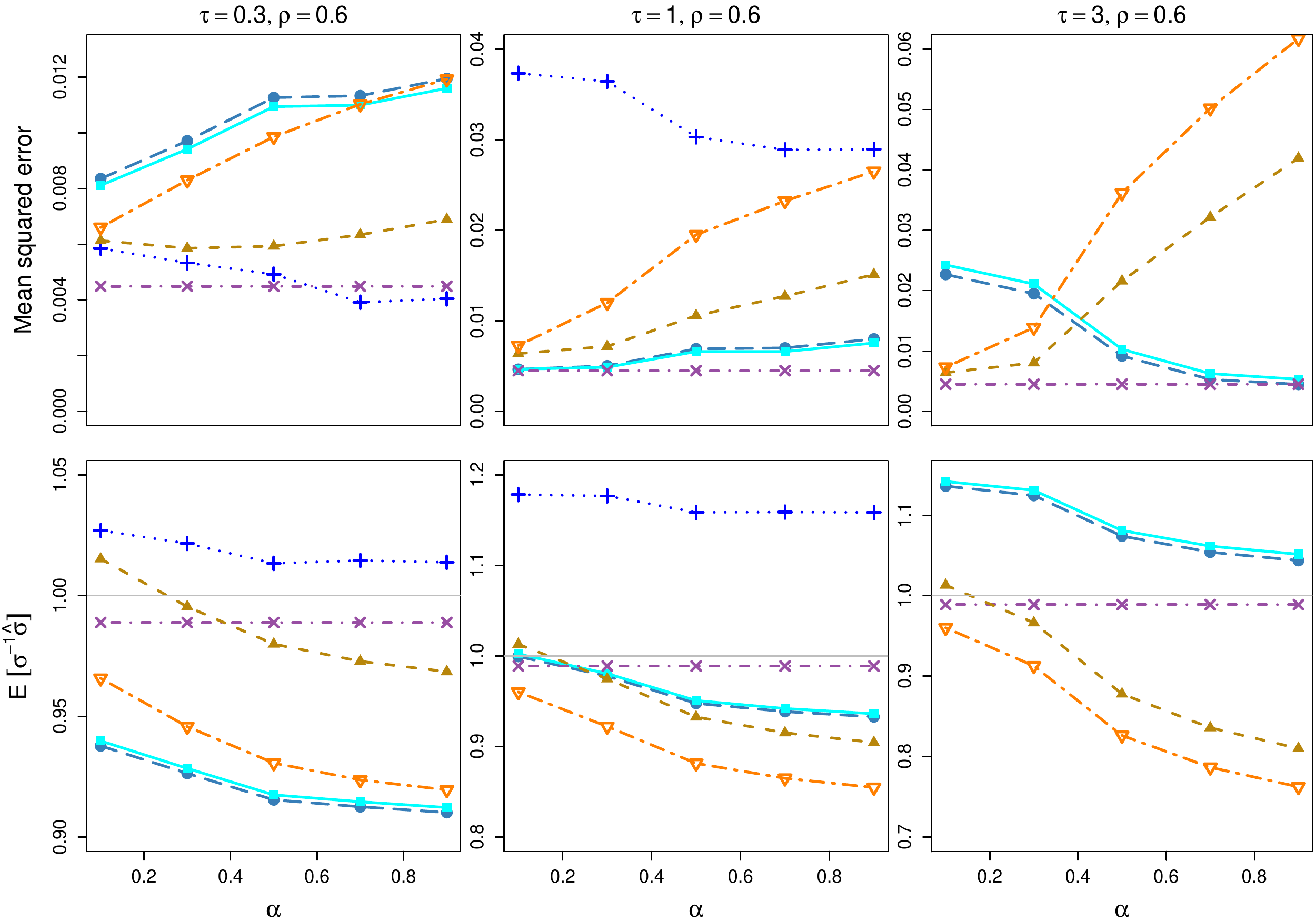}
    \caption{Simulation results of various methods with pre-specified
      regularization parameter values. From left to right, column show
      the average (over 1000 repetitions) of the mean squared error (top
      panel) and $E(\sigma^{-1} \hat{\sigma} )$ (bottom panel) of various
      methods in three simulation settings. In each setting, we fix
      model sparsity ($\alpha$) and correlations among features
      ($\rho$), and let signal-to-noise ratio(as expressed in $\tau$)
      change. 
      Line styles and their corresponding methods:
      \protect\includegraphics[width = 0.12in]{deep_blue.png} for organic ($\lambda_0$),
    \protect\includegraphics[width = 0.12in]{cyan.png} for organic ($\lambda_2$), 
    \protect\includegraphics[width = 0.12in]{blue.png} for organic ($\lambda_3$), 
    \protect\includegraphics[width = 0.12in]{brown.png} for scaled(1),
    \protect\includegraphics[width = 0.12in]{orange.png} for scaled (2), 
  \protect\includegraphics[width = 0.12in]{purple.png} for the oracle.}
  \label{fig:fix_appendix}
\end{figure}

\begin{table}[!htp]
  \caption{p-values for testing the difference of various methods outputs}
    \centering
    \begin{tabular}{lccc}
      \hline
    & natural vs. organic & $\rt$ vs. organic & $\rt$ vs. natural \\ 
    \hline
    $\alpha = 0.1, \rho = 0.3, \tau = 1$ & 0.00 (0.00) & 0.07 (0.00) & 0.00 (0.00) \\ 
    $\alpha = 0.3, \rho = 0.3, \tau = 1$ & 0.00 (0.00) & 0.19 (0.25) & 0.00 (0.00) \\ 
    $\alpha = 0.5, \rho = 0.3, \tau = 1$ & 0.00 (0.00) & 0.00 (0.00) & 0.00 (0.00) \\ 
    $\alpha = 0.7, \rho = 0.3, \tau = 1$ & 0.00 (0.00) & 0.00 (0.00) & 0.00 (0.00) \\ 
    $\alpha = 0.9, \rho = 0.3, \tau = 1$ & 0.00 (0.00) & 0.00 (0.00) & 0.00 (0.00) \\ 
    $\alpha = 0.1, \rho = 0.6, \tau = 1$ & 0.00 (0.00) & 0.08 (0.01) & 0.00 (0.00) \\ 
    $\alpha = 0.3, \rho = 0.6, \tau = 1$ & 0.00 (0.00) & 0.00 (0.14) & 0.00 (0.00) \\ 
    $\alpha = 0.5, \rho = 0.6, \tau = 1$ & 0.05 (0.10) & 0.01 (0.00) & 0.00 (0.00) \\ 
    $\alpha = 0.7, \rho = 0.6, \tau = 1$ & 0.00 (0.00) & 0.00 (0.00) & 0.00 (0.00) \\ 
    $\alpha = 0.9, \rho = 0.6, \tau = 1$ & 0.00 (0.00) & 0.00 (0.00) & 0.00 (0.00) \\ 
    $\alpha = 0.1, \rho = 0.9, \tau = 1$ & 0.06 (0.32) & 0.00 (0.03) & 0.00 (0.12) \\ 
    $\alpha = 0.3, \rho = 0.9, \tau = 1$ & 0.96 (0.02) & 0.00 (0.07) & 0.00 (0.00) \\ 
    $\alpha = 0.5, \rho = 0.9, \tau = 1$ & 0.03 (0.00) & 0.00 (0.00) & 0.00 (0.00) \\ 
    $\alpha = 0.7, \rho = 0.9, \tau = 1$ & 0.44 (0.00) & 0.00 (0.00) & 0.00 (0.00) \\ 
    $\alpha = 0.9, \rho = 0.9, \tau = 1$ & 0.20 (0.00) & 0.00 (0.01) & 0.00 (0.00) \\ 
    \hline
    \end{tabular}
  \label{tab:pvalue_t}
  \\[10pt]
  \caption*{
   In each simulation setting, as characterized by a $(\alpha, \rho, \tau)$ triplet,
 we report p-values of the (two-sided) paired t-tests and the Wilcoxon signed-rank tests (shown in parentheses) for testing the null hypothesis that the output of each pair of methods are the same.}
\end{table}
\begin{table}[!htp]
\def~{\hphantom{0}}
\caption{$E(\sigma^{-1} \hat{\sigma})$ in MSD dataset}
    \centering
    \begin{tabular}{lcccccc}
      \hline
      n & 20 & 40 & 60 & 80 & 100 & 120 \\ 
      \hline
      naive & ~80.1 (1.1) & ~94.2 (0.9) & ~95.8 (0.7) & ~96.4 (0.6) & ~97.9 (0.5) & ~96.7 (0.5) \\ 
      $\rt$ & ~90.0 (1.0) & 100.4 (0.8) & 101.7 (0.6) & 102.3 (0.5) & 103.3 (0.5) & 102.4 (0.4) \\ 
      natural & ~94.0 (0.9) & 103.3 (0.7) & 105.5 (0.6) & 106.0 (0.5) & 107.0 (0.4) & 106.6 (0.4) \\ 
      organic & ~86.8 (0.8) & ~97.6 (0.6) & ~99.9 (0.5) & 100.9 (0.4) & 101.7 (0.4) & 101.8 (0.4) \\ 
      scaled(1) & 106.1 (0.8) & 109.3 (0.6) & 111.2 (0.5) & 111.2 (0.4) & 111.7 (0.4) & 111.8 (0.4) \\ 
      scaled(2) & ~88.5 (0.8) & ~99.0 (0.6) & 102.9 (0.5) & 104.4 (0.5) & 105.1 (0.4) & 105.5 (0.3) \\ 
      organic($\lambda_2$) & ~89.7 (0.7) & ~94.7 (0.5) & ~97.6 (0.4) & ~98.3 (0.4) & ~99.2 (0.3) & ~99.7 (0.3) \\ 
      organic($\lambda_3$) & ~92.0 (0.7) & ~97.3 (0.6) & 100.1 (0.4) & 100.7 (0.4) & 101.6 (0.4) & 102.0 (0.3) \\ 
      \hline
    \end{tabular}
  \label{tab:MSD_ratio}
  \\[10pt]
  \caption*{
    Mean and standard errors (over 1000 replications) of $E(\sigma^{-1}\hat{\sigma})$ of various methods we considered in Section \ref{sec:simulation}. Each entry of the method output is multiplied by 100 to convey information more compactly.}
\end{table}

\end{appendices}
\clearpage
\bibliographystyle{agsm}
\bibliography{short.bib}

\end{document}